\pdfoutput=1

\documentclass[11pt]{article}
\usepackage[T1]{fontenc}
\usepackage{amssymb, amsmath, graphicx, subfigure, xparse}
\usepackage{float}
\usepackage{physics}
\usepackage{fullpage}
\usepackage[hidelinks]{hyperref}
\usepackage{authblk}
\usepackage{qcircuit}
\usepackage[dvipsnames]{xcolor}
\usepackage{mdframed}
\usepackage{mathpazo}

\usepackage{thmtools,mathtools}
\usepackage{thm-restate}

\usepackage{tikz}
\usetikzlibrary{shapes,backgrounds}
\usepackage{tikzsymbols}

\usetikzlibrary{decorations.pathmorphing}
\usetikzlibrary{decorations.pathreplacing}

\tikzset{snake it/.style={decorate, decoration=snake}}

\newtheorem{theorem}{Theorem}
\newtheorem{definition}[theorem]{Definition}

\newtheorem{lemma}[theorem]{Lemma}
\newtheorem{conjecture}[theorem]{Conjecture}

\newtheorem{claim}[theorem]{Claim}

\newenvironment{proof}{\noindent{\bf Proof:} \hspace*{1mm}}{
    \hspace*{\fill} $\Box$ }

\usepackage{amsmath, amssymb, graphicx, subfigure}
\usepackage{enumitem}
\usepackage{xparse}
\usepackage{physics}
\usepackage{color}

\newcommand{\CC}{\mathbb{C}}

\newcommand{\NN}{\mathbb{N}}
\newcommand{\II}{\mathbb{I}}

\newcommand{\Hh}{\mathcal{H}}

\newcommand{\Oo}{\mathcal{O}}
\newcommand{\Rr}{\mathcal{R}}

\newcommand{\supp}{\mathrm{supp}}

\newcommand{\eps}{\epsilon}

\newcommand{\half}{\frac{1}{2}}

\newcommand{\defeq}{=}

\usepackage{complexity}

\newcommand{\Enc}{\mathrm{Enc}}
\newcommand{\Dec}{\mathrm{Dec}}

\newcommand{\junk}{\mathsf{junk}}
\newcommand{\unary}{\mathsf{unary}}
\newcommand{\clock}{\mathsf{clock}}
\newcommand{\ancilla}{\mathsf{ancilla}} 
\usepackage{xcolor}

\definecolor{grisframe}{gray}{1.00}
\definecolor{gristitleframe}{gray}{0.85}
\mdfsetup{backgroundcolor=grisframe,
 skipabove=12pt,
 skipbelow=6pt,
 leftmargin=0pt,
 rightmargin=0pt,
 innertopmargin=12pt,
 innerbottommargin=6pt,
 innerleftmargin=12pt,
 innerrightmargin=12pt,
 frametitleaboveskip=12pt,
 frametitlebelowskip=6pt,
 frametitlerule=true,
 frametitlebackgroundcolor=gristitleframe,
 splittopskip=2\topsep}

\begin{document}

\title{Approximate low-weight check codes and circuit lower bounds for noisy ground states}
\author{Chinmay Nirkhe}
\author{Umesh Vazirani}
\author{Henry Yuen}
\affil{Electrical Engineering and Computer Sciences \\ University of California, Berkeley \\ $\mathtt{\{\href{mailto:nirkhe@cs.berkeley.edu}{nirkhe},\href{mailto:vazirani@cs.berkeley.edu}{vazirani},\href{mailto:hyuen@cs.berkeley.edu}{hyuen}\}@cs.berkeley.edu}$}
\date{}

\maketitle

\begin{abstract}

The No Low-Energy Trivial States (NLTS) conjecture of Freedman and Hastings (Quantum Information and Computation 2014), which asserts the existence of local Hamiltonians whose low-energy states cannot be generated by constant-depth quantum circuits, identifies a fundamental obstacle to resolving the quantum PCP conjecture. Progress towards the NLTS conjecture was made by Eldar and Harrow (Foundations of Computer Science 2017), who proved a closely related theorem called No Low-Error Trivial States (NLETS). In this paper, we give a much simpler proof of the NLETS theorem and use the same technique to establish superpolynomial circuit size lower bounds for noisy ground states of local Hamiltonians (assuming $\mathsf{QCMA} \neq \mathsf{QMA}$), resolving an open question of Eldar and Harrow. We discuss the new light our results cast on the relationship between NLTS and NLETS. 

Finally, our techniques imply the existence of \textit{approximate quantum low-weight check (qLWC) codes} with linear rate, linear distance, and constant weight checks. These codes are similar to quantum LDPC codes except (1) each particle may participate in a large number of checks, and (2) errors only need to be corrected up to fidelity $1 - 1/\mathsf{poly}(n)$. This stands in contrast to the best-known stabilizer LDPC codes due to Freedman, Meyer, and Luo which achieve a distance of $O(\sqrt{n} \log^{1/4} n)$.

The principal technique used in our results is to leverage the Feynman-Kitaev clock construction to approximately embed a subspace of states defined by a circuit as the ground space of a local Hamiltonian. \end{abstract}

\renewcommand{\time}{\mathsf{time}}
\newcommand{\state}{\mathsf{state}}

\section{Introduction}

The quantum PCP conjecture \cite{QuantumNPsurvey, Aharonov:2013:GCQ:2491533.2491549} is a central open question in quantum complexity theory. To understand the statement, it is helpful to review the dictionary translating between classical constraint satisfaction problems (CSPs) and their quantum analogue, the local Hamiltonian problem. A classical CSP on $n$ variables corresponds to a local Hamiltonian $H = H_1 + \cdots + H_m$ acting on $n$ qubits\footnote{For normalization, we assume that the terms of a local Hamiltonian have spectral norm at most $1$.}. A solution to the CSP corresponds to an $n$ qubit quantum state, and the number of violated constraints corresponds to the energy (eigenvalue) 
of that quantum state. The $\NP$-hardness of SAT corresponds to the $\QMA$-hardness of deciding whether the $H$ has minimum eigenvalue at most $a$ or at least $b$ for given $a, b$ such that $b - a = 1/\poly(n)$. The quantum analogue of the PCP theorem, called the qPCP conjecture, asserts that the problem remains $\QMA$-hard even when $b - a \geq c m = c \|H \|$. 

Just as the classical PCP theorem connects coding theory to constraint satisfaction problems, it is natural to expect any resolution of the quantum PCP conjecture to rely on --- and to reveal --- deep connections between the theory of quantum error-correcting codes and ground states (i.e. states of minimum energy) of local Hamiltonians. %
Examples of quantum error-correcting codes realized as the ground spaces of local Hamiltonians already play a central role in our understanding of the physical phenomenon known as topological order~\cite{KITAEV20032}. Moreover, it has been suggested that the qPCP conjecture is closely related to one of the biggest open problems in quantum coding theory: whether quantum low density parity check (qLDPC) codes with linear rate and linear distance exist \cite{freedman2002z2,bravyi2010tradeoffs,tillich2014quantum}.

The difficulty of the qPCP conjecture motivated Freedman and Hastings to formulate a simpler goal called the \emph{No Low-Energy Trivial States (NLTS) Conjecture}~\cite{DBLP:journals/qic/FreedmanH14}. 
One way to put one's finger on the additional difficulty of qPCP (beyond the ``standard'' difficulty of proving a classical PCP theorem) is that solutions of QMA-hard problems are expected to have high description complexity. For example, if $\NP \neq \QMA$, then ground states of local Hamiltonians do not have classically checkable polynomial-size descriptions.
The NLTS conjecture isolates this aspect of high description complexity by asserting the existence of a family of local Hamiltonians $\{H^{(n)}\}_{n=1}^\infty$ where $H^{(n)}$ acts on $n$ particles, such that low-energy states (of energy less than $c \|H \|$) cannot be generated by quantum circuits of constant depth. A much stronger version of the NLTS conjecture is a necessary consequence of the qPCP conjecture: assuming $\QCMA \neq \QMA$,\footnote{For precise definitions of the complexity classes $\QCMA$ and $\QMA$, we refer the reader to Section~\ref{preliminaries}. Roughly speaking, $\QMA$ is the class of problems for which the solution is a quantum state that can be efficiently checked by a quantum computer. $\QCMA$ is the class of problems where the solution is a \emph{classical} string that can be efficiently checked by a quantum computer.} low-energy states cannot be described even by polynomial-size quantum circuits. However, one of the advantages of the NLTS conjecture is that it does not involve complexity classes such as $\QMA$, but rather focuses on the entanglement complexity that is intrinsic to low-energy states of local Hamiltonians.

Like the qPCP conjecture, the NLTS conjecture remains unresolved. In~\cite{DBLP:conf/focs/EldarH17}, Eldar and Harrow proposed a variant of the NLTS called \emph{No Low-Error Trivial States (NLETS)}, which is itself a necessary consequence\footnote{The local Hamiltonian family must be of bounded-degree, meaning no particle participates in more than a constant number of Hamiltonian terms.} of NLTS. The difference was that rather than considering low-energy states of $H$, they considered a notion of ``local corruption error'', what they call \emph{$\eps$-error states}: these are states that differ from the ground state in at most
$\eps n$ qubits. More precisely, $\sigma$ is $\eps$-error for a local Hamiltonian $H$ if there exists a ground state $\rho$ of $H$ and a set $S$ of at most $\eps n$ qudits such that $\Tr_S(\rho) = \Tr_S(\sigma)$. Under this definition they were able to establish a family of Hamiltonians for which any $\eps$-error state requires circuit depth of $\Omega(\log n)$. This was welcomed as very encouraging progress 
towards establishing NLTS, since NLETS could be regarded as a close proxy for NLTS, with a technical change in definition of distance under which to examine the robustness of the ground space.

In this paper, we start by giving a simple argument for the $\Omega(\log n)$ circuit depth lower bound of Eldar and Harrow; our lower bound holds even under a more general error model, which allows any probabilistic mixture of $\eps$-error states (we call these states \emph{noisy ground states}). Moreover, we can use the same techniques to answer their open question of whether one can obtain circuit size lower bounds on low-error states that go beyond logarithmic depth: specifically, we show that there exists a family of local Hamiltonians whose noisy ground states require superpolynomial-size circuits, assuming $\QCMA \neq \QMA$.

One way to view these results is that they provide further progress towards the NLTS conjecture and beyond. %
However, it is instructive to take a step back to consider more closely the basic difference between NLETS and NLTS. This lies in the different notion of approximation: in NLETS, approximation corresponds to local corruptions in $\eps n$ sites, where $n$ is the total number of particles, whereas in NLTS approximation corresponds to energy at most $\eps \| H \|$ (intuitively, at most $\eps$ fraction of the terms of the Hamiltonian are violated). An alternative perspective on our results is that they suggest these two notions of approximation are quite different. This view is reinforced by the fact that our $\Omega(\log n)$-circuit depth lower bounds on noisy ground states holds for a family of 1D Hamiltonians, whereas we know that NLTS and qPCP Hamiltonians cannot live on any constant dimensional lattice~\cite{Aharonov:2013:GCQ:2491533.2491549}. This suggests that in the context of the qPCP and NLTS conjectures, the correct notion of distance is given by the energy or number of violated terms of the Hamiltonian.

On the other hand, the local corruption distance as defined by Eldar and Harrow for their NLETS result is the natural one that arises in quantum error correction: the distance of a code is \emph{defined} by the maximum number of qubits of a codeword that can be erased while maintaining recoverability. 
We give a construction of a family of codes (inspired by the construction used in our noisy ground state lower bound) that we call \emph{quantum low weight check (qLWC) codes}. The family of codes we consider are approximate error-correcting codes in the sense of~\cite{Crepeau:2005:AQE:2154598.2154621,beny2010general}. They are closely 
related to qLDPC codes, with the difference that they are not stabilizer codes and therefore the low weight checks are not Pauli operators. Specifically, we give a family of 
approximate qLWCs with linear distance and linear rate. Constructing qLDPC codes with similar parameters is a central open question in coding theory, with the 
best-known stabilizer LDPC codes due to Freedman, Meyer, and Luo which achieve a distance of $O(\sqrt{n} \log^{1/4} n)$ \cite{freedman2002z2}.

What is common to the above results is the technique. We start with the observation that the complicated part of the Eldar and Harrow proof is constructing a local Hamiltonian 
whose ground states share some of the properties of the cat state $\ket*{\Cat_n} = (\ket{0}^{\otimes n} + \ket{1}^{\otimes n})/\sqrt{2}$.  To do so, they constructed a local Hamiltonian corresponding to a quantum error-correcting code (based on the Tillich-Zemor hypergraph product construction) and showed that its ground states have non-expansion properties similar to those
of the cat state~\cite{DBLP:conf/focs/EldarH17}. Our starting point is the observation that the Tillich-Zemor construction is unnecessary and that one can make the cat state \emph{approximately} a ground state of a local Hamiltonian in the following sense: we construct a Feynman-Kitaev clock Hamiltonian corresponding to the circuit that generates $\ket*{\Cat_n}$ from $\ket{0}^{\otimes n}$.\footnote{A similar construction of a clock Hamiltonian was also considered by Crosson and Bowen in the context of idealized adiabatic algorithms \cite{Crosson2017QuantumGS}.} The ground state of this Hamiltonian is the history state of this computation, and we directly argue that the circuit depth necessary to generate this history state is at least $\Omega(\log n)$. This same argument even allows us to lower bound the circuit depth of \emph{approximate} noisy ground states (i.e. states that are close in trace distance to a noisy ground state).

The Feynman-Kitaev clock Hamiltonian plays a central role in our construction of qLWCs, with history states playing the role of codewords. 
The fact that such a construction yields an error-correcting code flies in the face of classical intuition. After all, it is the brittleness of the 
Cook-Levin tableau \cite{Cook:1971:CTP:800157.805047,Levin:73} (the classical analogue of the history state) that motivates the elaborate classical PCP constructions \cite{Arora:1998:PCP:273865.273901,Arora:1998:PVH:278298.278306,Dinur:2007:PTG:1236457.1236459}. The difference is 
that \emph{time} is in superposition in a quantum history state. We do not yet understand the implications of this observation. For example, is it possible that it might lead to new ways 
of constructing qLDPC codes with super-efficient decoding procedures? There are precedents for such connections between computational phenomena and codes, most notably with the PCP theorem and the construction of locally testable and locally checkable codes. %

Furthermore, while quantum error-correcting codes have typically provided a wealth of examples of interesting local Hamiltonians, our construction of qLWCs also suggest that a fruitful connection exists in the opposite direction: by considering techniques to construct local Hamiltonians (such as the Feynman-Kitaev clock construction), we can construct an interesting example of a quantum error-correcting code. We note that this reverse connection is starting to take hold in other areas of quantum information theory and physics: see \cite{brandao2017quantum,kim2017entanglement}.

\section{Summary of Results}

Before we present our results, we motivate our definition of \emph{noisy ground states}. 

\subsection{Noisy ground states}

The NLETS Theorem and NLTS conjecture describe different ways in which the ground space entanglement is robust. The ground states of NLETS Hamiltonians are robust against local corruptions in $\eps n$ sites, where $n$ is the total number of particles. NLTS Hamiltonians are robust against low-energy excitations in the sense that all states with energy at most $\eps \| H \|$ retain nontrivial circuit complexity. 

In this paper, we study another way that ground space entanglement can be robust. We focus on the concept of \emph{noisy ground state}, which is a generalization of low-error states: an $\eps$-noisy ground state $\sigma$ of a local Hamiltonian $H$ is a probabilistic mixture of $\eps$-error states $\{ \sigma_i \}$.

This notion of noisy ground state is naturally motivated by the following situation: consider a ground state $\rho$ of $H$. On each particle independently apply the following process $\mathcal{M}$: with probability $\eps$, apply a noisy channel $\mathcal{N}$, and with probability $1 - \eps$ apply the identity channel $\mathcal{I}$. The resulting state is
\begin{equation}
\begin{aligned}
	\mathcal{M}(\rho) &= \left ( (1 - \eps) \mathcal{I} + \eps \mathcal{N} \right)^{\otimes n}(\rho) \\
	&= \sum_{S \subseteq [n]} (1 - \eps)^{n - |S|} \eps^{|S|} \mathcal{N}^S(\rho) \\
	&\approx \sum_{S: |S| \leq 2 \eps n} (1 - \eps)^{n - |S|} \eps^{|S|} \mathcal{N}^S(\rho)
\end{aligned}
\end{equation}
where $\mathcal{N}^S$ denotes the tensor product of the map $\mathcal{N}$ acting on the particles indexed by $S$. The last approximate equality follows from the fact that with overwhelmingly large probability, $\mathcal{N}^S$ acts on at most $2 \eps n$ particles. Notice that the expression on the right hand side is (up to normalization) a $2\eps$-noisy ground state, because when $|S| \leq 2 \eps n$, the state $\mathcal{N}^S(\rho)$ is a $2 \eps$-error state.

This justifies the name ``noisy ground state'', as the operation $\mathcal{M}$ is is a reasonable model of noise that occurs in physical processes (and is frequently considered in work on quantum fault-tolerance). Furthermore, we believe that our model arises naturally in the context of noisy adiabatic quantum computation. 

As mentioned before, noisy ground states are a generalization of low-error states but are a special case of low-energy states: since low-error states are themselves low-energy states, a convex combination of them is also low-energy.

\vspace{20pt}

\noindent We prove several results about the robustness of entanglement in noisy ground states.

\subsection{Logarithmic circuit depth lower bound} 

First, we generalize Eldar and Harrow's logarithmic circuit depth lower bound \cite{DBLP:conf/focs/EldarH17} to encompass noisy ground states. Furthermore, we present a family of Hamiltonians that is \emph{one dimensional}; in other words, the particles of the Hamiltonian are arranged on a line and the Hamiltonian terms act on neighboring particles.

We call this the \emph{Logarithmic Noisy Ground States (LNGS) Theorem}\footnote{We pronounce this ``Longs.''}.

\begin{restatable}[Logarithmic lower bound]{theorem}{lngs}
There exists a family of $3$-local Hamiltonians $\{ H^{(n)} \}$ on a line, acting on particles of dimension $3$, such that for all $n \in \NN$, for all $0 \leq \eps < 0.11$, $0 \leq \delta < \frac{7}{16} - 12\eps + \sqrt{8\eps}$, the $\delta$-approximate circuit depth of any $\eps$-noisy ground state $\sigma$ for $H^{(n)}$ is at least $\half \log_2 (n/2)$.
\label{theorem-lngs}
\end{restatable}

Here, the $\delta$-approximate circuit depth of $\rho$ means the circuit depth needed to produce a state that is $\delta$-close to $\rho$ in trace distance.

Our proof of Theorem~\ref{theorem-lngs} is simple and self-contained. As a consequence of our simpler local Hamiltonian construction, we obtain improved parameters over those in \cite{DBLP:conf/focs/EldarH17}. Furthermore, as we will discuss below in Section~\ref{sec:implications}, the fact that our LNGS Hamiltonian is one dimensional gives a strong separation between NLETS/LNGS and NLTS Hamiltonians.

\subsubsection{Superpolynomial circuit size lower bound} 
A question that was left open by~\cite{DBLP:conf/focs/EldarH17} is whether one can obtain circuit lower bounds on low-error states that are better than logarithmic -- say polynomial or even exponential. We show that there exists a family of local Hamiltonians whose noisy ground states require superpolynomial\footnote{Here, ``superpolynomial'' refers to functions $f(n)$ that grow faster than any polynomial in $n$.} size circuits, assuming $\QCMA \neq \QMA$. Since low-error states are noisy ground states, this provides an answer to Eldar and Harrow's open question. 

We call this the \emph{Superpolynomial Noisy Ground States (SNGS) Theorem}\footnote{We pronounce this ``Songs''.}.

\begin{restatable}[Superpolynomial Noisy Ground States (SNGS)]{theorem}{superpoly}
If $\QCMA \neq \QMA$, then there exists $q, \eps > 0$ and a family of $7$-local Hamiltonians $\{ H^{(n)} \}$ acting on dimension-$q$ qudits such that for all $0 \leq \delta < 1/5$, the $\delta$-approximate circuit complexity of any family $\{ \sigma_n \}$ of $\eps$-noisy ground states for $\{ H^{(n)} \}$ grows faster than any polynomial in $n$.
\label{theorem-superpoly}
\end{restatable}

We call such a family $\{H^{(n)} \}$ \emph{SNGS Hamiltonians}. The following is a proof sketch. Let $L = (L_{yes}, L_{no})$ be the $\QMA$-complete language consisting of descriptions of polynomial-size verifier circuits acting on a witness state and ancilla qubits. We convert each circuit $C \in L$, into a circuit $C'$ where $C'$  applies in order: (a) a unitary $V$ to encode the state in an error-correcting code, (b) a collection of identity gates, (c) the unitary $V^\dagger$ to decode the state, and (d) the gate circuit $C$. The construction maintains that the circuits $C'$ and $C$ are equivalent. We then generate the Feynman-Kitaev clock Hamiltonian for $C'$. Let $H_C$ be this Hamiltonian. The family of SNGS Hamiltonians is precisely $\{H_C : C \in L_{yes}\}$.

In order to prove that all noisy ground states of this Hamiltonian must have superpolynomial circuit size, we show that if there was a noisy ground state with a polynomial-size generating circuit, then the description of the generating circuit would suffice as a \emph{classical} witness for the original $\QMA$-complete problem. In the {\em yes} case, the construction of $C'$ from $C$ enforces that tracing out the $\sf{time}$ register of the noisy ground state will yield a state close to a convex combination of $\{\Enc(\ket{\xi_i,0})\}$ where $\Enc(\cdot)$ is the encoding function for the error-correcting code and $\{\ket{\xi_i}\}$, a collection of accepting witness. Therefore, given the description of the generating circuit for the noisy ground state, we can generate the noisy ground state and decode the original witness state. It suffices then to check the witness by running the original circuit $C$. The {\em no} case follows easily from the definition of $L_{no}$. This proves that $L \in \QCMA$, proving $\QCMA = \QMA$, contradicting the original assumption. 

\subsubsection{Semi-explicit SNGS Hamiltonians via oracle separations}

It is an open question in quantum complexity theory of whether $\QCMA$ is equal to $\QMA$. Aaronson and Kuperberg gave the first complexity-theoretic evidence that they are different by constructing a \emph{quantum} oracle $\Oo$ such that $\QCMA^\Oo \subsetneq \QMA^\Oo$~\cite{aaronson2007quantum}. Fefferman and Kimmel later showed that one can obtain the same oracle separation with \emph{in-place} oracles $\Oo$, which are permutation matrices in the standard basis~\cite{DBLP:journals/corr/FeffermanK15}. The separations of~\cite{aaronson2007quantum,DBLP:journals/corr/FeffermanK15} hold as long as the locality of the oracles $\Oo$ is $\omega(\log n)$ (i.e. superlogarithmic in the problem size).

We show that any oracle separation between $\QCMA$ and $\QMA$ can be leveraged to obtain a semi-explicit family of SNGS Hamiltonians:

\begin{restatable}{theorem}{superpolyoracle}
\label{theorem-superpoly-oracle}
There exists $q, \eps > 0$, a function $k(n) = O(\log^{1 + \alpha} n)$ for arbitrarily small $\alpha > 0$ and a family of $k$-local Hamiltonians $\{ H^{(n)} \}$ acting on dimension-$q$ qudits such that the following holds: 
The circuit complexity of any family $\{ \sigma_n \}$ of $\eps$-noisy ground states for $\{ H^{(n)} \}$ grows faster than any polynomial in $n$. Furthermore, there is exactly one term in $H^{(n)}$ that is $k(n)$-local; all other terms are $7$-local. 
\end{restatable}

Unlike Theorem~\ref{theorem-superpoly}, the superpolynomial lower bound on the circuit complexity of noisy ground states does not require any complexity-theoretic assumption! The caveat is that this family is only known to exist via a counting argument; there is exactly \emph{one} term of the Hamiltonian that has $\omega(\log n)$-locality and does not have an explicit description. However, however, all of other the terms of the local Hamiltonians are $7$-local and have explicit descriptions.

The essential idea is to apply the proof of Theorem~\ref{theorem-superpoly} to the $\QMA^\Oo$ verifier that decides a language $L$ which is not in $\QCMA^\Oo$. In both~\cite{aaronson2007quantum,DBLP:journals/corr/FeffermanK15}, this verifier only makes a single call to the oracle $\Oo$. Thus there is one term in the Feynman-Kitaev clock Hamiltonian corresponding to the propagation of that oracle call. Since we do not have an explicit description of a separating oracle $\Oo$, this Hamiltonian term is non-explicit.

\newcommand{\wt}[1]{\widetilde{#1}}

\subsection{Asymptotically good approximate low-weight check codes}

The techniques from the previous sections also give rise to what we call \emph{approximate quantum low-weight check (qLWC) codes}. These are closely related to quantum low-density parity check (qLDPC) codes, which are stabilizer codes where each parity check acts on a bounded number of particles, and each particle participates in a bounded number of parity checks. It is a long-standing open question of whether asymptotically good qLDPC codes exist (i.e. constant locality, constant rate, and constant relative distance). The qLDPC conjecture posits that such codes exist.

We show that if one relaxes the conditions of (a) each particle participating in a small number of constraints, and (b) that we can \emph{exactly} recover from errors, we can obtain locally defined quantum error-correcting codes with such good parameters. First, we define our notion of approximate qLWC codes:

\begin{definition}[Approximate qLWC code]
	A local Hamiltonian $H = H_1 + \cdots + H_m$ acting on $n$ dimension-$q$ qudits is a $[[n,k,d]]_q$ \emph{approximate quantum LWC code} with error $\delta$ and locality $w$ iff each of the terms $H_i$ act on at most $w$ qudits and there exists encoding and decoding maps $\Enc, \Dec$ such that
	\begin{enumerate}
		\item $\bra{\Psi}H\ket{\Psi} = 0$ if and only if $\ketbra{\Psi}{\Psi} = \Enc(\ketbra{\xi}{\xi})$ for some $\ket{\xi} \in (\CC^q)^{\otimes k}$.
		\item For all $\ket{\phi} \in (\CC^q)^{\otimes k} \otimes \mathcal{R}$ where $\mathcal{R}$ is some purifying register, for all completely positive trace preserving maps $\mathcal{E}$ acting on at most $(d-1)/2$ qudits,
	\begin{equation}
		\left \| \Dec \circ \mathcal{E} \circ \Enc(\ketbra{\phi}{\phi}) - \ketbra{\phi}{\phi} \right \|_1 \leq \delta.
	\end{equation}
	Here, the maps $\Enc$, $\mathcal{E}$, and $\Dec$ do not act on register $\mathcal{R}$.
	\end{enumerate}
\end{definition}
The first condition of the above definition enforces that the ground space of the Hamiltonian $H$ of an approximate qLWC code is a $q^k$-dimensional codespace; it is the exactly the image of the encoding map $\Enc$. The second condition corresponds to the approximate error-correcting condition, where we only require that the decoded state is \emph{close} to the original state (i.e., we no longer insist that $\Dec \circ \mathcal{E} \circ \Enc$ is exactly the identity channel $\mathcal{I}$). Although there are few results on approximate quantum error-correcting codes, we do know that relaxing the exact decoding condition yields codes with properties that cannot be achieved using exact codes~\cite{leung1997approximate,Crepeau:2005:AQE:2154598.2154621,beny2010general}.

Our proof of Theorem~\ref{theorem-superpoly} yields a construction of an approximate quantum LWC code with distance $\Omega(n)$, and we give a self-contained presentation of approximate qLWC code constructions in Section~\ref{sec-good-approx-qldpc-codes}. We believe this may be of independent interest.

\begin{restatable}[Good approximate qLWC codes exist]{theorem}{goodapproxqldpc}
	For all error functions $\delta(n)$ there exist a family of $[[n,k,d]]_q$ approximate quantum LWC codes 
	with the following parameters:
	\begin{center}
	\begin{tabular}{ r c c l }
	 Qudit dimension &$q$& $=$ &$O(1)$,\\ 
	 Error & $\delta$ & $=$ & $\delta(n)$,\\  
	 Locality & $w$ & $=$ & $3 + 2r$,\\
	 Blocklength & $n$ & $=$ & $O(rk)$, \\
	 Distance & $d$ & $=$ & $\Omega(n/r)$   
	\end{tabular}
	\end{center}
	where
	\begin{equation}
		r = O \left ( \frac{\log (1 + 4/\delta^2)}{\log n} \right) + 2.
	\end{equation}
	Furthermore, the encoding and decoding maps for these codes are explicit and efficiently computable.	
\end{restatable}
Observe that when $\delta(n) = 1/\poly(n)$, the parameter $r = O(1)$.

By comparison, the best-known qLDPC codes (of the stabilizer variety) with constant locality have distance bounded by $O(\sqrt{n} \log^{1/4} n)$~\cite{freedman2002z2}. Hastings constructs a qLDPC stabilizer code with constant locality that has distance $n^{1 - \eps}$ for any $\eps > 0$, assuming a conjecture in high dimensional geometry~\cite{Hastings17,hastings2016weight}. Bacon, et al. were able to construct sparse subsystem codes (a generalization of stabilizer codes) with constant locality and distance $n^{1 - o(1)}$~\cite{bacon2017sparse}. We note that, interestingly, the codes of~\cite{bacon2017sparse} are constructed from fault-tolerant quantum circuits that implement a stabilizer code --- this is similar to the way we construct our approximate qLWC codes!

\subsection{Implications for NLTS, quantum PCP and quantum LDPC}
\label{sec:implications}

Our investigation into noisy ground states and approximate low-weight check codes is motivated by a number of important open questions in quantum information theory: NLTS, quantum PCP, and quantum LDPC. We believe that our results help clarify the status of these open problems, and the relationships between them.

\paragraph*{A separation between LNGS/SNGS and NLTS Hamiltonians.} First, our logarithmic circuit-depth lower bound for noisy ground states (Theorem~\ref{theorem-lngs}) gives a strong separation between the notions of entanglement robustness in NLETS and NLTS: we showed that a one-dimensional local Hamiltonian is NLETS. However, it is easy to see that one-dimensional Hamiltonians (or any Hamiltonian on a constant-dimensional lattice) cannot be NLTS. To see this, consider taking a $n$-particle ground state $\ket{\Psi}$ of a 1D Hamiltonian $H$; divide up the $n$ particles into contiguous chunks of length $L$. Let $\sigma = \rho_1 \otimes \rho_2 \otimes \cdots \otimes \rho_{n/L}$ where $\rho_i$ is the reduced density matrix of $\ket{\Psi}$ on the $i$'th chunk. This state $\sigma$ violates $O(n/L)$ terms of the Hamiltonian (since $H$ is one-dimensional). Therefore, it is a $\eps$-energy state of $H$ for $L = \Theta(1/\eps)$. On the other hand, $\sigma$ is a tensor product state that can be generated by $2^{O(1/\eps)}$-depth circuits, which is constant for constant $\eps$. This indicates that the form of entanglement robustness as expressed in NLETS and in our LNGS/SNGS Hamiltonian constructions is much weaker than the entanglement robustness required by the NLTS conjecture and quantum PCP, where one has to look for Hamiltonians on high dimensional geometries.

\paragraph*{Quantum LDPC codes and the Quantum PCP conjecture.} Resolving the qPCP conjecture would likely involve a transformation from $H$ to $H'$ that (at the very least) has the property that exact ground states of $H$ (or closeby states in trace distance) can be recovered from \emph{low-energy} states of $H'$. It has been suggested that such a transformation would involve some kind of qLDPC code~\cite{DBLP:journals/qic/FreedmanH14,doi:10.1137/140975498,Hastings17,DBLP:conf/focs/EldarH17}. In fact, it is believed that a special kind of qLDPC code, called a \emph{quantum locally testable code (qLTC)}, is necessary~\cite{doi:10.1137/140975498}. However, the existence of qLTCs (or even qLDPC codes) with constant relative distance is a major open problem. 

We believe our results on approximate quantum LWC codes present two take-home messages for the qPCP and qLDPC conjectures. First, it is important that a qPCP (or a qLTC) Hamiltonian be local, but it is \emph{not} necessary that the Hamiltonian be bounded degree (meaning that each particle only participates in a few terms). The bounded degree condition is useful in the original context for qLDPCs, where an important motivation is to find fast decoding algorithms. In the context of qPCP/qLTC, however, decoding efficiency is not an immediate concern; thus resolving the qPCP conjecture \emph{need not} resolve the qLDPC conjecture.

Second, we believe this gives evidence that considering codes other than stabilizer codes --- such as approximate codes or subsystem codes --- may be useful in the quest for both qPCP and qLDPC. Most work on qLDPC codes has focused on constructing CSS and stabilizer codes, but it may be fruitful to branch out beyond the CSS/stabilizer setting for the purposes of understanding the possibilities (or limits) of qPCP/qLDPC. For example, our qLWC codes are unconventional in a few ways: they are defined by non-commuting Hamiltonians, they only admit approximate recovery, and each particle participates in a large number of checks.

\subsection{Open questions}

We list a few open problems.
\begin{enumerate}
\item Are there SNGS Hamiltonians or (approximate) qLWC codes that are geometrically local (with respect to, say, the Euclidean metric)? Our construction of a 1-dimensional NLGS Hamiltonian uses a simplification of a technique of Aharanov et. al. \cite{Aharonov:2007:AQC:1328722.1328726} of converting a quantum circuit into a 2-dimensional local Hamiltonian. This technique works because of the specific structure of the circuit generating the $\ket*{\Cat}$ state. In general, the transformation involves increasing the number of qudits by more than a constant factor. If this factor is $\Theta(n^\alpha)$, then the ground states are resilient to errors of size at most $n^{1-\alpha}$.

\item Is there a family of local Hamiltonians such that any \emph{superposition} (not just convex combination) of low-error states have large circuit complexity? This notion is a generalization of a noisy state; such states have small \emph{quantum Hamming distance} to the ground space. This is an interesting notion in the context of quantum locally testable codes (qLTCs) because low-energy states are equivalent to states with low quantum Hamming distance to the codespace (see~\cite{DBLP:conf/focs/EldarH17} for definitions of quantum Hamming distance and qLTCs). 

\item Are there applications of our qLWC constructions?

\item There has been a number of recent results about approximate quantum error-correcting codes in a variety of areas including many-body physics~\cite{brandao2017quantum}, the AdS/CFT correspondence~\cite{kim2017entanglement}, and quantum resource theories~\cite{hayden2017approximate}. Could approximate error-correcting codes play a role in trying to resolve the qPCP and qLDPC conjectures?

\item Eldar and Harrow showed that quantum locally testable codes of the CSS type are NLTS~\cite{DBLP:conf/focs/EldarH17}. Can this argument be extended to general qLTCs?

\item Is it possible for qLDPC codes (not necessarily stabilizer or exact error-correcting codes) to be defined as the codespace of a geometrically local Hamiltonian? There are a few no-go results that give limitations on codes living on lattices~\cite{bravyi2010tradeoffs,Flammia2017limitsstorageof}, but they apply to special classes of codes such as stabilizer codes or locally-correctible codes. Our qLWC codes, by contrast, are neither. A follow-up work to this one, shows that this is indeed possible modulo polylogarithmic corrections \cite{approxqldpc}. 

\item Could the \emph{combinatorial NLTS conjecture} be easier to prove than the NLTS conjecture? This conjecture posits that there exist a family of local Hamiltonians where states that have non-zero energy penalty on only a small constant fraction of Hamiltonian terms must have non-trivial circuit complexity. 
\end{enumerate}

\subsection*{Outline}

In Section \ref{preliminaries}, we provide definitions and formal statements of the conjectures introduced above. In Section \ref
{simple-proof}, we prove the LNGS theorem using our techniques. In Section \ref{qcmaneqqma} we prove the superpolynomial variant (SNGS) assuming $\QCMA \neq \QMA$. In Section \ref{unconditionalsuperpoly}, we prove the unconditional version holding for $O(\log^{1+\kappa} n)$-local Hamiltonians. In Section \ref{sec-good-approx-qldpc-codes}, we prove that asymptotically good approximate qLWC codes exist.

\section*{Acknowledgments}
We thank Dorit Aharanov, Itai Arad, Adam Bouland, Elizabeth Crosson, Bill Fefferman, Lior Eldar, Zeph Landau, Ashwin Nayak, Nicholas Spooner, and Thomas Vidick for helpful discussions. We also thank  Ali Lavasani, Lisa Yang, and anonymous referees for pointing out errors in an earlier version of this work. This work was supported by ARO Grant W911NF-12-1-0541 and NSF Grant CCF-1410022.

\section{Preliminaries}
\label{preliminaries}

We will assume that the reader is familiar with the basics of quantum computing and quantum information.

\subsection{Quantum Merlin-Arthur}

\begin{definition}[$\QMA$]
A quantum circuit $C$ acting on $n$ qubits is a $\QMA$-verifier circuit iff there exists $m \leq n$ qubits that are designated the $\sf witness$ register and the rest of the qubits form the $\sf ancilla$ register, and it satisfies the promise that either there exists a state $\ket{\xi} \in (\CC^2)^{\otimes m}$ such that
\begin{equation}
 	\Pr(C \text{ accepts } \ket{\xi,0}) \geq 2/3
\end{equation}
or for all states $\ket{\xi}$,
\begin{equation}
 	\Pr(C \text{ accepts } \ket{\xi,0}) \leq 1/3.
\end{equation}
By accept, we mean the event that measuring the first qubit of the state $C \ket{\xi,0}$ in the standard basis yields the $\ket{1}$ state.
\end{definition}

\begin{definition}[$\QCMA$]
A quantum circuit $C$ acting on $n$ qubits is a $\QCMA$-verifier circuit iff there exists $m \leq n$ qubits that are designated the $\sf witness$ register and the rest of the qubits form the $\sf ancilla$ register, and it satisfies the promise that either there exists a witness string $w \in \{0,1\}^m$ such that
\begin{equation}
 	\Pr(C \text{ accepts } \ket{w,0}) \geq 2/3
\end{equation}
or for all strings $w \in \{0,1\}^m$,
\begin{equation}
 	\Pr(C \text{ accepts } \ket{w,0}) \leq 1/3.
\end{equation}
\end{definition}
The constants $2/3$ and $1/3$ are arbitrary; they only need to be separated by a universal constant. 

\subsection{Quantum PCP and related problems}

Here we give formal definitions of the quantum PCP conjectures.

\begin{conjecture}[Quantum PCP Conjecture \cite{QuantumNPsurvey}]
It is $\QMA$-hard to decide whether a given local Hamiltonian $H = H_1 + \cdots + H_m$ (where each $\| H_i\| \leq 1$) has minimum eigenvalue at most $a$ or at least $b$ when $b - a \geq c \|H \|$ for some universal constant $c > 0$.
\end{conjecture}

\begin{conjecture}[The NLTS Conjecture \cite{DBLP:journals/qic/FreedmanH14}]
There exists a universal constant $\eps > 0$ and an explicit family of local Hamiltonians $\{H^{(n)}\}_{n=1}^\infty$ where $H^{(n)}$ acts on $n$ particles and consists of $m_n$ local terms, such that any family of states $\{\ket{\psi_n} \}$ satisfying $\expval{H^{(n)}}{\psi_n} \leq \eps \| H^{(n)} \| + \lambda_{\min}(H^{(n)})$ requires circuit depth that grows faster than any constant.
\end{conjecture} 

\begin{theorem}[NLETS Theorem~\cite{DBLP:conf/focs/EldarH17}]
\label{theorem-eldar-harrow}
	There exists a family of $16$-local Hamiltonians $\{H^{(n)} \}$ such that any family of $\eps$-error states $\{ \ket{\Phi_n} \}$ for $\{ H^{(n)} \}$ requires circuit depth $\Omega(\log n)$, where $\eps = 10^{-9}$.
\end{theorem}

\subsection{Circuits}
\label{sec-prelim-circuits}

\begin{definition}[Circuit depth/size]
	Let $U$ be a unitary acting on $(\CC^q)^{\otimes n}$ such that $U = U_m \cdots U_1$ where each $U_i$ is a unitary acting on at most two qudits (called a \emph{gate}). We say that $U$ has \emph{circuit size $m$}, and has \emph{circuit depth $d$} if there exists a partition of $\{ U_i \}$ into $d$ \emph{layers} $L_1,\ldots,L_d$ where each layer $L_j$ is a set of non-overlapping two-qudit unitaries and
		\begin{equation}
			U = \left ( \bigotimes_{i \in L_d} U_i \right) \cdots \left ( \bigotimes_{i \in L_1} U_i \right).
		\end{equation}
		In other words, $U$ can be written as a product of $d$ layers of a tensor product of disjoint two-local unitaries.
\end{definition}

\paragraph{Lightcones.} 
Let $U = L_d \cdots L_1$ be a depth-$d$ circuit acting on $(\CC^q)^{\otimes n}$, where each $L_j = \bigotimes_i U_{ji}$ is a tensor product of disjoint two-qudit unitaries $U_{ji}$. Let $A$ be an operator.

Define $K^{(d)}$ to be the set of two-qudit gates in layer $d$ whose supports overlap with that of $A$. Now for every $j = d-1,\ldots,1$, define $K^{(j)}$ to be the set of two-qudit gates in layer $j$ whose supports overlap with any gate in $K^{(j+1)},\ldots,K^{(d)}$. The \emph{lightcone} of the operator $A$ with respect to $U$ is defined as the union of these sets:
\begin{equation}
	K \defeq \bigcup_j K^{(j)}.
\end{equation}
In other words, the lightcone of $A$ is the set of gates emanating from $A$ to the first layer of the circuit. Furthermore, we write $\supp(K)$ to denote the set of qudits that are touched by the lightcone of $A$. Observe that if $A$ acts on a single qudit, then $|\supp(K)| \leq 2^d$.

\begin{claim}
\label{clm:lightcone}
	Let $U$ be a depth-$d$ circuit with two-local gates, and let $A$ be a operator acting on a single qudit. The number of qudits whose associated lightcones intersect the lightcone of $A$ is at most $2^{2d + 1}$.	
\end{claim}
\begin{proof}
Define $E^{(1)} = K^{(1)}$, and define $E^{(2)}$ to be the set of two-qudit gates in layer $2$ whose supports overlap with $K^{(1)}$. For $j = 3,\ldots,d$, define $E^{(j)}$ to be the set of two-qudit gates in layer $j$ whose supports overlap with any gate in $E^{(j-1)}$. Define the \emph{effect zone} of the operator $A$ with respect to $U$ to be the union of these sets:
\begin{equation}
	E_U(A) \defeq \bigcup_j E^{(j)}.
\end{equation}
In other words, the effect zone of an operator $A$ is essentially the ``bounceback of the lightcone'' of $A$: i.e., the set of gates emanating from the first layer $K^{(1)}$ (the ``widest part'' of the lightcone) to the last layer of the circuit. 

Now finally define the \emph{shadow of the effect zone} of $A$ with respect to $U$ to be the set $W_U(A)$ of qudits that are acted on by the gates in $E_U(A)$. Since $K^{(1)}$ has at most $2^d$ gates, $|E_U(A)| \leq 2^d |K^{(1)}| \leq 2^{2d}$. Therefore the size of the shadow is at most $2^{2d + 1}$ because each gate can act on at most $2$ qudits. 

It follows that the shadow of the effect zone are all the qudits whose lightcones could intersect the lightcone of $A$.

\end{proof}

\subsection{States and complexity}
We will use $\|A \|_1$ to denote the trace norm $\Tr(\sqrt{A A^\dagger})$ of an operator $A$. In addition, $D(\Hh)$ denotes the space of positive semidefinite operators of trace norm 1 on $\Hh$.
\begin{definition}[Approximate circuit depth/approximate circuit complexity]
	A state $\rho$ has $\delta$-approximate circuit depth $D$ (resp. $\delta$-approximate circuit size $S$) iff there exists a state $\sigma$ such that $\| \rho - \sigma \|_1 \leq \delta$ and the circuit depth of $\sigma$ is $D$ (resp. the circuit size of $\sigma$ is $S$).
\end{definition}

\begin{definition}[Low error states]
	Let $\rho,\sigma \in D((\CC^d)^{\otimes n})$, and let $H$ be a local Hamiltonian acting on $(\CC^d)^{\otimes n}$. Then
	\begin{enumerate}
	 	\item We say that $\sigma$ is an $\eps$-error state of $\rho$ if there exists a subset $S \subseteq [n]$ of size at most $\eps n$ such that $\Tr_S(\rho) = \Tr_S(\sigma)$;
		\item We say that $\sigma$ is an $\eps$-error state for $H$ if there exists a state $\rho$ such that $\Tr(H \rho) = \lambda_{\min}(H)$ (i.e. is a ground state) and $\sigma$ is an $\eps$-error state of $\rho$.
	\end{enumerate}
\end{definition}

\begin{definition}[Noisy ground states]
A state $\rho \in D((\CC^d)^{\otimes n})$ is an $\eps$-noisy ground state of a local Hamiltonian $H$ acting on $(\CC^d)^{\otimes n}$ if $\rho$ can be expressed as a convex combination of $\eps$-error states for $H$. Equivalently, there exists a set of $\eps$-error states $\{\rho_i\}$ for $H$ and a probability distribution $\{p_i\}$ over them such that $\rho = \sum_i p_i \rho_i$.
\end{definition}

\begin{definition}[Unary clock]
For all $T \in \mathbb{N}$ and $t \in \{0, \ldots, T\}$ we define the $T$-qudit state $\ket{\unary(t,T)}$ to be
\begin{equation}
\ket{\unary(t,T)} = \ket{0}^{\otimes (T-t)} \otimes \ket{1}^{\otimes t}.
\end{equation}

\end{definition}

Some of our results will require a more succinct clock where we write the time as a coordinate in a $k$-dimensional cube of volume $T$.
\begin{definition}[$k$-dimensional clock]
For all $k,T \in \mathbb{N}$, let $d = \lceil T^{1/k} \rceil$ + 1. For all $t \in \{0, \ldots, T\}$, let $a_1,\ldots,a_k \in \{0,\ldots,d - 1\}$ be the unique solutions to $t = a_k d^{k-1} + \cdots + a_1$. We define the $k(d-1)$ qudit state $\ket{\clock_k(t,T)}$ to be %
\begin{equation}
\ket{\clock_k(t,T)} = \bigotimes_{i=1}^k \ket{\unary(a_i,d)}.%
\end{equation}
\end{definition}
Note that for $k = 1$, $\ket{\clock_k(t,T)} = \ket{\unary(t,T)}$. Furthermore, throughout this paper we generally will not specify the value of $T$. It will be assumed to be the minimal size necessary to express $\ket{\unary(t)}$ and $\ket{\clock_k(t)}$ for all $t$ involved in the analysis.

\begin{definition}[History state]
Let $C$ be a quantum circuit that acts on two registers, $\sf witness$ and $\sf ancilla$. Let $C_1,\ldots,C_T$ denote the sequence of two-local gates in $C$. Then for all $k \in \mathbb{N}$ a state $\ket{\Psi}\in \Hh_{\sf{time}} \otimes \Hh_{\sf{state}}$ is a \emph{$k$-dimensional history state of $C$} if
\begin{equation}
\ket{\Psi} = \frac{1}{\sqrt{T+1}} \sum_{t = 0}^T \ket{\clock_k(t)}_{\sf{time}} \otimes \ket{\psi_t}_{\sf{state}}
\end{equation}
where $\ket{\psi_0}_{\sf state} = \ket{\xi}_{\sf witness} \otimes \ket{0}_{\sf ancilla}$ for some state $\ket{\xi}$ and $\ket{\psi_t} = C_t \ket{\psi_{t-1}}$ for $t = 1,\ldots,T$. 

A state $\ket{\Psi}$ is a \emph{history state} if there is some $k$ and a circuit $C$ for which it is a $k$-dimensional history state for $C$.
\end{definition}
In this paper we will repeatedly invoke the following simple, but useful, Lemma.
\begin{lemma}
Let $\ket{\Psi} = \frac{1}{\sqrt{T+1}} \sum_{t=0}^T \ket{\clock_k(t)}_\time \otimes \ket{\psi_t}_\state$ be a history state on $m$ qudits and let $\ket{\Phi}$, $S \subseteq [m]$ be be such that $\Tr_S(\ketbra{\Phi}{\Phi}) = \Tr_S(\ketbra{\Psi}{\Psi})$. Then
\begin{equation}
\Tr_{S \cup \time}(\ketbra{\Phi}{\Phi}) = \frac{1}{T+1} \sum_{t=0}^T \Tr_{S \setminus \time}(\ketbra{\psi_t}{\psi_t}).
\end{equation}
\label{lemma-traceorder}
\end{lemma}

\begin{proof}
Tracing out $S$ followed by $\time \setminus S$ is equivalent to tracing out $\sf{time}$ followed by $S \setminus \sf{time}$. Therefore,
\begin{align}
\Tr_{S \cup \sf{time}}(\ketbra{\Phi}{\Phi}) &= \Tr_{S \cup \sf{time}}(\ketbra{\Psi}{\Psi}) \\
&= \Tr_{S \setminus \sf{time}}( \Tr_{\sf{time}}(\ketbra{\Psi}{\Psi})) \\
&= \Tr_{S \setminus \sf{time}} \left( \frac{1}{T+1} \sum_{t = 0}^T \ketbra{\psi_t}{\psi_t} \right) \\
&= \frac{1}{T+1}  \sum_{t = 0}^T  \Tr_{S \setminus \sf{time}} \left( \ketbra{\psi_t}{\psi_t} \right).
\end{align}
\end{proof}

\subsection{The Feynman-Kitaev clock construction}
\label{sec-clock-construction}

\subsubsection{Unary clock}

In this section, we review the Feynman-Kitaev clock construction~\cite{kitaev2002classical} and show that low-energy states of the clock Hamiltonian are close to unary history states. 

Let $C$ be a $\QMA$-verifier circuit acting on $n$ qubits\footnote{This generalizes to qudits in a straightforward way.} comprised of a sequence of two-local gates $C_1,\ldots,C_T$, with $T = \poly(n)$. The verifier circuit takes in input a witness $\ket{\xi}$ on $m \leq n$ qubits, and $n - m$ ancilla $\ket{0}$ states. The promise is that either
\begin{itemize}
\item (\textbf{YES} case) There exists a witness state $\ket{\xi}$ such that 
\begin{equation}
	\Pr \left ( \text{$C$ accepts $\ket{\xi}$} \right ) \geq 1 - \gamma
\end{equation}
or
\item (\textbf{NO} case) For all states $\ket{\xi}$
\begin{equation}
	\Pr \left ( \text{$C$ accepts $\ket{\xi}$} \right ) \leq \gamma.
\end{equation}
\end{itemize}
Through standard amplification techniques for $\QMA$, we will assume without loss of generality that $\gamma \leq \exp(-n^a)$ for some $a> 0$. The Feynman-Kitaev Hamiltonian corresponding to the circuit $C$ is defined to be the sum
\begin{equation}
H = H_\text{in} + H_\text{prop} + H_\text{out} + H_\text{stab}
\end{equation} 
acting on $\Hh = \Hh_{\time} \otimes \Hh_\state$ where $\Hh_\time = (\CC^{2})^{\otimes T}$ and $\Hh_\mathsf{state} = (\CC^2)^{\otimes n}$. The terms are defined as follows. The Hamiltonian $H_\text{prop} = \sum_{t=1}^T H_t$ enforces that the ground state is a history state of the circuit $C$, where
\begin{equation}
\label{eq-prop-term}
H_t = \half W_t \left ( \ket{10}_{t,t+1} - \ket{11}_{t,t+1} \right) \left ( \bra{10}_{t,t+1} - \bra{11}_{t,t+1} \right ) W_t^\dagger \otimes \ketbra{0}{0}_{t+2}
\end{equation}
The subscripts $t,t+1,t+2$ are shorthand for $\time(t),\time(t+1),\time(t+2)$, respectively. The operator $W_t$ is defined as
\begin{equation}
W_t = \sum_{t = 1}^T \ketbra{\unary(t)}{\unary(t)}_\time \otimes C_t C_{t-1} \cdots C_1.
\end{equation}
The term $H_\text{in}$ enforces that the history state initializes the last $n - m$ qubits of the computation in the all zeroes state:
\begin{equation}
	H_\text{in} = \ketbra{0}{0}_{\time(1)} \otimes \sum_{s = m+1}^n \ketbra{1}{1}_{\state(s)}.
\end{equation}
The term $H_\text{out}$ enforces that the history state encodes an accepting computation:
\begin{equation}
	H_\text{out} = \ketbra{1}{1}_{\time(T)} \otimes \ketbra{1}{1}_{\state(1)}. 
\end{equation}
Finally, the term $H_\text{stab}$ enforces that the $\time$ register is supported only on valid unary encodings:
\begin{equation}
	H_\text{stab} = \sum_{t = 1}^{T-1} \ketbra{01}{01}_{\time(t) \time(t+1)}.
\label{eq-stab-term}
\end{equation}

The Hamiltonian $H$ is $5$-local. The problem of estimating the minimum eigenvalue of $H$ to inverse polnyomial accuracy was proved to be $\QMA$-hard in~\cite{kitaev2002classical}.
\begin{theorem}[$\QMA$-completeness of Local Hamiltonians~\cite{kitaev2002classical}]
If there exists a witness state $\ket{\xi}$ such that $\Pr(C \text{ accepts } \ket \xi) \geq 1 - \gamma$ (i.e. $C$ is a {\em yes} instance), then 
\begin{equation}
	\lambda_{\min}(H) \leq \gamma/(T+1).
\end{equation}
If for all states $\ket{\xi}$ we have $\Pr(C \text{ accepts } \ket \xi) \leq \gamma$ (i.e. $C$ is a {\em no} instance), then
\begin{equation}
	\lambda_{\min}(H) \geq c(1-\sqrt{\gamma})T^{-3}
\end{equation}
 for some constant $c$.
\label{theorem-kitaev}
\end{theorem}

To prove the \emph{yes} case, one shows that a history state that encodes the computation of the verifier circuit $C$ on input $\ket{\xi,0}$ has energy at most $O(\gamma/T)$. In the next theorem, we show that in the \emph{yes} case \emph{all} low-energy states of $H$ are close to some unary history state.

\begin{theorem}
Let $\ket{\eta}$ be a state such that $\expval{H}{\eta} \leq \delta$. Then there exists a history state $\ket{\Psi}$ such that $\norm{\ket{\eta} - \ket{\Psi}} \leq \poly(T) \sqrt{\delta}$.
\label{theorem-closegroundstates}
\end{theorem}

\begin{proof}
Let $H'$ be the partial Hamiltonian $H_\text{in} + H_\text{prop} + H_\text{stab}$. As every term of $H$ is postive semidefinite, $\expval{H'}{\eta} \leq \delta$ as well. It is easy to see that $H'$ has a ground energy of 0 and the ground space is
\begin{equation}
G = \left\{ \ket{\Psi} : \ket{\Psi} \text{ a history state} \right\}.
\end{equation} 
It is well known that the spectral gap $\Delta$ of $H'$ is $1/\poly(T)$ (see, e.g.,~\cite{Aharonov:2007:AQC:1328722.1328726}). We then appeal to the Gentle Measurement Lemma~\cite{ogawa2002new} to prove that 
\begin{equation}
\norm{\ketbra{\eta}{\eta} - \ketbra{\Psi}{\Psi}}_1 \leq \sqrt{\delta / \Delta} = \poly(T) \sqrt{\delta}
\end{equation} where $\ket{\Psi} = \Pi \ket{\eta} / \sqrt{\expval{\Pi}{\eta}}$ for $\Pi$ the projector onto $G$.
\end{proof}

\subsubsection{$k$-dimensional clock}
Kitaev's original clock construction is 5-local because each term of the Hamiltonian checked at most 3 qubits of the $\time$ register and at most 2 qubits of the $\state$ register. This, however, required that the length of the $\time$ register was equal to the size of the circuit. We modify this construction by instead storing the time as a $\ket{\clock_k(\cdot)}$ state. The resulting Hamiltonian will be $2k + 3$-local as now $2k + 1$ registers will be required to keep track of the $k$ dimensions of the $\time$ register. It is not difficult to see that the proof of Theorem \ref{theorem-kitaev} translates directly to $k$-dimensional clocks. 

We provide a full description of the Hamiltonian in Appendix \ref{appendix-kitaevclock} for the interested reader. The proofs in this paper will not require knowledge of the explicit construction. %

\section{A simple proof of Logarithmic noisy ground states}
\label{simple-proof}

\lngs*

The first step in the proof is to create the 5-local Hamiltonian corresponding to a circuit generating the cat state $\ket*{\Cat_n}$. This Hamiltonian has a logarithmic circuit lower-bound but is not spatially local. We then transform the Hamiltonian into a one-dimensional system of qutrits with $3$-local interactions. \\

\begin{proof}
Fix an $n \in \NN$. Consider the circuit in Figure~\ref{fig:cat-ckt}.
\begin{figure}
\centering
$
\Qcircuit @C=.7em @R=.4em @! {
\lstick{\ket{0}} & \gate{H} & \ctrl{1} & \qw & \qw & \qw &\\
\lstick{\ket{0}} & \qw & \targ & \ctrl{1} & \qw & \qw& \\
\lstick{\ket{0}} & \qw & \qw   & \targ  & \ctrl{1} & \qw & \cdots \\
\lstick{\ket{0}} & \qw & \qw   & \qw 	& \targ  & \qw	& \\
\qquad \qquad \qquad \qquad \qquad \vdots \\
\vspace{10em}
}
$
\caption{Circuit $C$ to generate the $\ket*{\Cat_n}$ state.}
\label{fig:cat-ckt}
\end{figure}
When extended appropriately, this circuit generates the cat state, $\ket*{\Cat_n} = \frac{1}{\sqrt{2}} \left ( \ket{0}^{\otimes n} + \ket{1}^{\otimes n} \right)$. %

We first start with the $5$-local unary clock Feynman-Kitaev Hamiltonian $H = H_\text{in} + H_\text{prop} + H_\text{stab}$ based on $C$, where $H_\text{prop}$ and $H_\text{stab}$ are as in Section~\ref{sec-clock-construction} and we have a term that checks that the input to the circuit $C$ is initialized to all zeroes:
\begin{equation}
	H_\text{in} = \sum_{i = 1}^n \Pi_{\time(1)}^{(1)} \otimes \Pi_{\state(i)}^{(1)}.
\end{equation}

The ground energy of $H$ is $0$ and, furthermore, it has a unique ground state, which is the history state $\ket{\Psi}$ corresponding to running circuit $C$ on input $\ket{0}^{\otimes n}$:
\begin{equation}
\ket{\Psi} = \frac{1}{\sqrt{n + 1}} \sum_{t = 0}^{n} \ket{\unary(t)}_\time \otimes \ket{\psi_t}_\state = \frac{1}{\sqrt{n+1}} \sum_{t = 0}^{n} \ket{\unary(t)}_\time \otimes \ket*{\Cat_t} \otimes \ket{0}^{\otimes (n-t)}.
\end{equation}

However, this Hamiltonian is not geometrically local, so we now apply a transformation inspired by that of Aharonov et. al. ~\cite{Aharonov:2007:AQC:1328722.1328726} to embed the Hamiltonian on a one dimensional lattice (i.e. a line). Suppose we overlay the $n$ qubits of the $\time$ register onto of the $n$ qubits of the $\state$ register, and now fuse the $\time(i)$ qubit with the $\state(i)$ qubit into a single particle of dimension $4$. When viewed as acting on this space of dimension-$4$ particles $\time(i) \otimes \state(i)$, the Hamiltonian $H$ is geometrically local whose geometry is the line, since the gates of circuit $C$ act only between consecutive neighboring $\state$ qubits. 
In particular, we note that term $H_t$ of $H_\text{prop}$ acts solely on qubits $\time(t-1), \time(t), \time(t+1), \state(t-1), \state(t)$, so $H$ is $3$-geometrically local.

\begin{figure}[h!]
\centering
\begin{tikzpicture}[scale = 1.2];

\draw (1, 1) ellipse (0.25 and 0.25);
\node[align=center] at (1,1) {$1$};

\draw (2, 1) ellipse (0.25 and 0.25);
\node[align=center] at (2,1) {$1$};

\draw (3, 1) ellipse (0.25 and 0.25);
\node[align=center] at (3,1) {$1$};

\draw (4, 1) ellipse (0.25 and 0.25);
\node[align=center] at (4,1) {$0$};

\draw (5, 1) ellipse (0.25 and 0.25);
\node[align=center] at (5,1) {$0$};

\draw [decorate,decoration={brace,amplitude=3pt}]
(6,1.25) -- (6,0.75) node [black,midway,xshift=20pt] 
{$\time$};

\draw (1, 0) ellipse (0.25 and 0.25);
\node[align=center] at (1,0) {$\Cat$};

\draw[snake it] (1.25,0) -- (1.75,0);

\draw (2, 0) ellipse (0.25 and 0.25);
\node[align=center] at (2,0) {$\Cat$};

\draw[snake it] (2.25,0) -- (2.75,0);

\draw (3, 0) ellipse (0.25 and 0.25);
\node[align=center] at (3,0) {$\Cat$};

\draw (4, 0) ellipse (0.25 and 0.25);
\node[align=center] at (4,0) {$0$};

\draw (5, 0) ellipse (0.25 and 0.25);
\node[align=center] at (5,0) {$0$};

\draw [decorate,decoration={brace,amplitude=3pt}]
(6,0.25) -- (6,-0.25) node [black,midway,xshift=20pt] 
{$\state$};

\draw[dashed, very thin] (0.65,1.35) -- (1.35,1.35) -- (1.35,-0.35) -- (0.65,-0.35) -- cycle;
\draw[dashed, very thin] (1.65,1.35) -- (2.35,1.35) -- (2.35,-0.35) -- (1.65,-0.35) -- cycle;
\draw[dashed, very thin] (2.65,1.35) -- (3.35,1.35) -- (3.35,-0.35) -- (2.65,-0.35) -- cycle;
\draw[dashed, very thin] (3.65,1.35) -- (4.35,1.35) -- (4.35,-0.35) -- (3.65,-0.35) -- cycle;
\draw[dashed, very thin] (4.65,1.35) -- (5.35,1.35) -- (5.35,-0.35) -- (4.65,-0.35) -- cycle;

\node[align=center] at (1,-1) {$\Downarrow$};
\node[align=center] at (2,-1) {$\Downarrow$};
\node[align=center] at (3,-1) {$\Downarrow$};
\node[align=center] at (4,-1) {$\Downarrow$};
\node[align=center] at (5,-1) {$\Downarrow$};

\draw (1, -2) ellipse (0.25 and 0.25);
\node[align=center] at (1,-2) {$\Cat$};

\draw[snake it] (1.25,-2) -- (1.75,-2);

\draw (2, -2) ellipse (0.25 and 0.25);
\node[align=center] at (2,-2) {$\Cat$};

\draw[snake it] (2.25,-2) -- (2.75,-2);

\draw (3, -2) ellipse (0.25 and 0.25);
\node[align=center] at (3,-2) {$\Cat$};

\draw (4, -2) ellipse (0.25 and 0.25);
\node[align=center] at (4,-2) {$2$};

\draw (5, -2) ellipse (0.25 and 0.25);
\node[align=center] at (5,-2) {$2$};

\end{tikzpicture}
\caption{Transformation of Hamiltonian from qubits to qutrits.}
\label{fig-registers-in-line}
\end{figure}
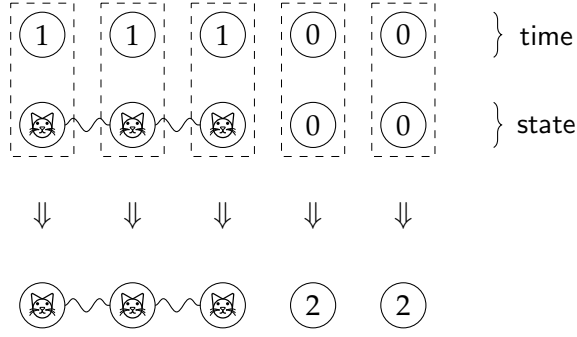

Finally, we observe that one of the degrees of freedom of the fused particle is unused, as $\bra{\Psi} \left( \ket{0}_{\time(i)} \otimes \ket{1}_{\state(i)} \right) = 0$. Therefore, we can group qubits $\time(i)$ and $\state(i)$ into a dimension-3 qudit (i.e. \emph{qutrit}) under the mapping
\begin{equation}
\ket{1}_{\time(i)} \ket{x}_{\state(i)} \mapsto \ket{x}_i \text{ for } x \in \{0,1\}, \qquad \ket{0}_{\time(i)} \ket{0}_{\state(i)} \mapsto \ket{2}_i.
\end{equation}
Figure \ref{fig-registers-in-line} shows how the qubits containing $\ket{\unary(3)}_\time \otimes (\ket*{\Cat_3} \otimes \ket{0}^{\otimes 2})_\state$ are mapped into qutrits.

The unique ground state of $H$, under this mapping, is now
\begin{equation}
\ket{\Psi} = \frac{1}{\sqrt{n + 1}} \sum_{t = 0}^{n} \ket{\psi_t} = \frac{1}{\sqrt{n+1}} \sum_{t = 0}^{n} \ket*{\Cat_t} \otimes \ket{2}^{\otimes (n-t)}.
\end{equation}

Let $\sigma$ be an $\eps$-noisy ground state of $\ket{\Psi}$, which is now a $n$ qutrit state. Then we can express $\sigma$ as $\sum_\ell p_\ell \sigma_\ell$ where each $\sigma_\ell$ is an $\eps$-error state of $\ketbra{\Psi}{\Psi}$. Let $S_\ell$ be the subset of qutrits of size at most $\eps n$ for which $\Tr_{S_\ell}(\sigma_\ell) = \Tr_{S_\ell}(\ketbra{\Psi}{\Psi})$.

We now make a simple observation. For $i \in [n]$, let $A_i = \ketbra{0}{0}_{i}$ be the projector onto the $\ket{0}$ state of the $i$th qutrit. Similarly, let $B_i = \ketbra{1}{1}_i$. Then we have that for all $\ell$, for all $i < j$ that are not in $S_\ell$, 
\begin{align}
\Tr( A_i \otimes B_j \sigma_\ell) &= \Tr( A_i \otimes B_j \Tr_{S_\ell}(\sigma_\ell)) \\
&= \Tr( A_i \otimes B_j \Tr_{S_\ell}(\ketbra{\Psi}{\Psi})) \\
&= \Tr( A_i \otimes B_j \ketbra{\Psi}{\Psi}) \\
&= \frac{1}{n+1} \left \| \sum_{t} A_i \otimes B_j \ket{\psi_t} \right \|^2 \\
&= 0.
\label{eq-oppositesignequalzero1}
\end{align}
This is because of the following reasoning: fix a $t$, and suppose $j > t$. Then the ``wave of CNOT'' gates has not reached the $j$'th qutrit yet, so the $j$'th qutrit of $\ket{\psi_t}$ is in the state $\ket{2}$, meaning that $B_j \ket{\psi_t} = 0$. On the other hand, suppose $j \leq t$; we have $A_i \otimes B_j \ket{\psi_t} = 0$, because the operator $A_i$ will collapse the cat state into the all zeroes state, which has zero overlap with the projector $B_j$. %

For each index $i \in \{1, \ldots, n\}$, call qutrit $i$ good if the probability (with respect to the distribution $\{ p_\ell \}$) that $i$ is in $S_\ell$ is at most $4\eps$.
By a simple counting argument, we know that there must be at least $(3/4)n$ good qutrits. %

Let $d < (1/2) \log_2(n/2) - 1$. Suppose for contradiction there existed a depth-$d$ two-local circuit $U$ that, after tracing out all but $n$ qutrits, generates a density matrix $\sigma'$ that is $\delta$-close to $\sigma$ in trace distance. Then since the lightcone of any qutrit $i$ intersects the lightcones of at most $2^{2d + 1} < n/2$ qutrits, there exists good qutrits $i \leq n/4$ and $j \leq (3/4)n$ whose lightcones do not intersect.
Claim~\ref{clm:lightcone} then implies that the lightcones of $A_i$ and $B_j$ are disjoint. The following Proposition, proved in Appendix \ref{appendix-prop-lightcone}, shows the measurements of disjoint observables on qutrits $i$ and $j$ have uncorrelated outcomes:
\begin{restatable}[]{proposition}{lightconeprop}
\label{prop-lightcone}
Let $\rho$ be a density matrix acting on an $n$-qudit state generated by a two-local quantum circuit $U$ after tracing out some qudits. Let $A,B$ be operators whose lightcones with respect to $U$ do not intersect. Then
\begin{equation}
	\Tr(A \otimes B \rho) = \Tr(A \rho) \cdot \Tr(B \rho).
\end{equation}
\end{restatable}
This Proposition implies
\begin{align}
	\Tr(A_i \otimes B_j \sigma) &\geq \Tr(A_i \otimes B_j \sigma') - \delta \\
							&= \Tr(A_i \sigma') \cdot \Tr(B_j \sigma') - \delta \\
							&\geq \left ( \Tr(A_i \sigma) - \delta \right) \left(\Tr(B_j \sigma) - \delta \right) - \delta \label{eq:prod-of-small-tr}
\end{align}
where in the equality we use the Proposition, and in the two inequalities we use the fact that $\| \sigma - \sigma' \|_1 \leq \delta$. 

Call $\ell$ \emph{good} if both $i \notin S_\ell$ and $j \notin S_\ell$. By a union bound, at least $1 - 8\eps$ fraction of the $\ell$ are good (with respect to the probability distribution $\{p_\ell\}$). Therefore,
\begin{align}
\Tr(A_i \otimes B_j \sigma) = \sum_\ell p_\ell \Tr(A_i \otimes B_j \sigma_\ell)
\leq 8\eps + \sum_{\text{good } \ell} p_\ell \Tr(A_i \otimes B_j \sigma_\ell)
= 8 \eps. \label{eq-oppositesignequalzero2}
\end{align}
The last equality follows from~\eqref{eq-oppositesignequalzero1}. We now lower bound $\Tr(A_i \sigma)$ and $\Tr(B_j \sigma)$.
\begin{align}
	\Tr(A_i \sigma) &= \sum_\ell p_{\ell} \Tr(A_i \sigma_\ell) \\
					&\geq \sum_{\text{good } \ell} p_\ell \Tr(A_i \sigma_\ell) - 4\eps \\
					&= \sum_{\text{good } \ell} p_\ell \Tr(A_i \ketbra{\Psi}{\Psi}) - 4\eps \\
					&\geq \frac{1}{n+1} \sum_{t,t'} \bra{\psi_{t'}} A_i \ket{\psi_t} - 12\eps \\
					&= \frac{1}{n+1} \sum_{t} \bra{\psi_{t}} A_i \ket{\psi_t} - 12\eps
\end{align}
where in the first inequality, we use the fact that $i$ is good for $1 - 4 \eps$ fraction of $\ell$'s; in the third line, we use the fact that $i \notin S_\ell$. In the fourth line, we use that the probability of good $\ell$ is at least $1 - 8\eps$. The last equality follows from the fact that $\bra{\psi_{t'}} A_i \ket{\psi_t} = 0$ whenever $t \neq t'$. 

If $i > t$, $\bra{\psi_t} A_i \ket{\psi_t} = 0$. If $i \leq t$, $\bra{\psi_{t}} A_i \ket{\psi_t} = 1/2$. Since $i \leq n/4$, we have $\Tr(A_i \sigma) \geq \frac{3}{8} - 12\eps$. Similarly, since $j \leq (3/4)n$, we have $\Tr(B_j \sigma) \geq \frac{1}{8} - 12\eps$.

Inputting these bounds into (\ref{eq:prod-of-small-tr}) and recalling the bound calculated in (\ref{eq-oppositesignequalzero2}), we have 
\begin{equation}
	 (1/8 - 12\eps - \delta)(3/8 - 12\eps - \delta) - \delta \leq \Tr(A_i \otimes B_j \sigma) \leq 8\eps,
\end{equation}
which is a contradiction for our choice of $\eps$ and $\delta$. Therefore, our assumption of $d < \half \log(n/2)$ is false, completing the proof.
\end{proof}

\section{Superpolynomial noisy ground states}

\label{qcmaneqqma}
\newcommand{\sR}{{\sf R}}

\superpoly* %

The proof proceeds as follows: we start with the fact that the local Hamiltonians problem is complete for the class $\QMA$. We argue that if all local Hamiltonians have noisy ground states that can be (approximately) generated by polynomial-sized circuits, then the local Hamiltonians problem can be solved in $\QCMA$: a classical description of a circuit for a noisy ground state of the local Hamiltonian can be used to certify, in quantum polynomial time, that a local Hamiltonian has a low energy ground state. Thus, if we assume that $\QCMA \neq \QMA$, this implies that there are families of local Hamiltonians with noisy ground states with superpolynomial circuit complexity.

\begin{proof}

\paragraph*{\textbf{A $\QMA$-complete language}.} Let $L = (L_{yes},L_{no})$ be the $\QMA$-complete language such that instances of $L$ are descriptions of verifier circuits that act on a witness state and some ancilla qubits. A circuit $C$ is in $L_{yes}$ if there exists a witness state $\ket{\xi}$ such that $C$ accepts $\ket{\xi,0}$ with probability at least $1 - \gamma$. A circuit $C$ is in $L_{no}$ if for all witness states, $C$ accepts $\ket{\xi,0}$ with probability at most $\gamma$. Via standard $\QMA$ amplification techniques, we will assume without loss of generality that $\gamma \leq \exp(-n^a)$ for some universal constant $a > 0$, where $n$ represents the number of qubits that $C$ acts on. Furthermore, through padding we can always assume that $C$ has fewer than $n$ gates.

\paragraph*{\textbf{A good error-correcting code}.} Let $\{ Q_k \}$ be an explicit family of quantum error-correcting codes that are asymptotically good. By explicit, we mean that there is a uniform family of polynomial-size circuits for encoding and decoding (which includes error-correction). By asymptotically good, we mean that each $Q_k$ is a $[[n,k,d]]$ code with $n = O(k)$ and $d = \Omega(n)$. An example of such a code is given in Appendix~\ref{sec:good_code}.

\paragraph*{\textbf{Getting an error-corrected circuit}.} We now describe a polynomial time  reduction that transforms an input $\QMA$-verification circuit $C$ into an equivalent circuit $C'$ such that deciding whether $C \in L_{yes}$ or $C \in L_{no}$ is equivalent to deciding whether $C'$ is a {\em yes} instance or a {\em no} instance. Let $k$ denote the number of qubits that $C$ acts on, and assume without loss of generality that the size of $C$ is at most $k$ (if not, then simply pad $C$ with unused ancilla qubits so that this is the case). Let $Q = Q_k$ be the $[[n,k,d]]$ code from the family above that encodes $k$ logical qubits into $n = \alpha k$ physical qubits and has distance $d = \beta n$ for some universal constants $\alpha, \beta > 0$.

Let $E$ and $D$ denote the encoding and decoding circuits for $Q$, respectively. For an $k$-qubit message $\ket{\psi}$, the unitary $E$ maps $\ket{\psi} \ket{0^{n - k}}$ to an $n$ qubit codeword. For an $n$-qubit codeword $\ket{\theta} = E (\ket{\psi} \ket{0^{n-k}})$ and any unitary error $U$ that acts on at most $(d-1)/2$ qubits, the unitary $D$ maps $(U \ket{\theta}) \ket{0^s}$ to $\ket{\psi} \ket{0^{n-k}} \ket{\tau_U}$. Here, $\ket{\tau_U}$ is some ``junk'' state that depends on $U$, and $s = O(n)$ is the number of ancilla qubits used by the decoding procedure.

Consider the following circuit $C'$. It acts on three registers $\sR_1 \sR_2 \sR_3$. The first register $\sR_1$ consists of $k$ qubits, and corresponds to the space acted upon by the circuit $C$. The second register $\sR_2$ consists of $n-k$ qubits, and corresponds to the ancillas used by the encoding procedure of the code $Q$. The third register $\sR_3$ consists of $s$ qubits, and corresponds to the ancillas used by the decoding procedure. In total, the circuit $C'$ acts on $O(n)$ qubits.

Let $T_E$, $T_D$, and $T_C$ denote the circuit sizes of $E$, $D$, and $C$ respectively. Let $K = T_E + T_D + T_C$. The structure of the circuit $C'$ is as follows: 
\begin{enumerate}
\item Apply the encoding circuit $E$ with message register $\sR_1$ and ancilla register $\sR_2$.
\item Apply the identity gate for $K$ time steps.
\item Apply the decoding circuit $D$ to registers $\sR_1 \sR_2 \sR_3$, where $\sR_1 \sR_2$ is treated as the codeword register, and $\sR_3$ is treated as the ancilla register.
\item Apply the circuit $C$ to register $\sR_1$.
\end{enumerate}
See Figure~\ref{fig:cprime} for a diagram of $C'$.
Observe that
\begin{equation}
	\frac{K}{T_E + K + T_D + T_C} = \half.
\end{equation}
In other words, the ``idle'' period of the circuit $C'$ is $50\%$ of the circuit size. Let $T_{enc} = T_E$ and $T_{dec} = T_E + K$. The waiting period of the circuit is from time steps $T_{enc}+1,\ldots,T_{dec}$. 
\\
\begin{figure}[h!]
\centering
$
\Qcircuit @C=2em @R=1em {
\lstick{} & \multigate{3}{\quad E\quad } & \qw & \multigate{5}{\qquad \II \qquad} & \qw & \multigate{5}{\quad D\quad } & \qw & \multigate{2}{\quad C\quad } & \qw \\
\lstick{} & \ghost{\quad E\quad } & \qw & \ghost{\qquad \II \qquad} & \qw & \ghost{\quad D\quad } & \qw & \ghost{\quad C\quad } & \qw \inputgroupv{1}{2}{.4em}{1.2em}{\ket{\psi} \quad }\\
\lstick{} & \ghost{\quad E \quad } & \qw & \ghost{\qquad \II \qquad} & \qw & \ghost{\quad D\quad } & \qw & \ghost{\quad C\quad } & \qw\\
\lstick{} & \ghost{\quad E\quad } & \qw & \ghost{\qquad \II \qquad} & \qw & \ghost{\quad D\quad } & \qw & \qw & \qw \inputgroupv{3}{4}{.4em}{1.2em}{\ket{0^{n-k}} \quad }\\
\lstick{} & \qw & \qw & \ghost{\qquad \II \qquad} & \qw & \ghost{\quad D\quad } & \qw & \qw & \qw\\
\lstick{} & \qw & \qw & \ghost{\qquad \II \qquad} & \qw & \ghost{\quad D\quad } & \qw & \qw & \qw 
\inputgroupv{5}{6}{.4em}{1.2em}{\ket{0^s} \quad }
}
$
\\
\caption{The circuit $C'$.}
\label{fig:cprime}
\end{figure}
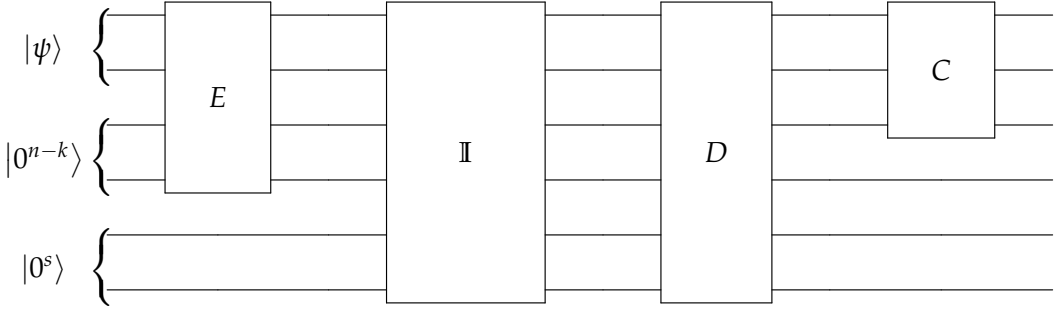 \ \\

\noindent Clearly, this circuit $C'$ is equivalent to $C$: it accepts a witness $\ket{\psi}$ with exactly the same probability as $C$ would. Let $c$ be a constant such that the circuit size $T_{C'}$ of $C'$ is at most $n^c$. 

\paragraph*{\textbf{Our family of local Hamiltonians.}} Let $C \in L_{yes}$ and let $C'$ denote the circuit obtained via the transformation described above. Let $H_C$ denote the Feynman-Kitaev clock Hamiltonian corresponding to $C'$ where clock is made $c$-dimensional, as described in Section~\ref{sec-clock-construction}. The Hamiltonian $H_C$ is $(2c+3)$-local and acts on $m = a cn$ qubits for some constant $a > 0$. Since $C$ is a {\em yes} instance, by Theorem~\ref{theorem-kitaev} the minimum energy of $H_C$ is at most $\gamma$. %

Our family of local Hamiltonians is defined to be the set $\{ H_C : C \in L_{yes}\}$.

\paragraph*{The lower bound.} We now prove the claimed circuit lower bound. For each $C\in L_{yes}$ let $\rho_C$ be an $\eps$-noisy ground state for $H_C$ for $\eps \leq \theta / 2\beta$. Assume for contradiction that the family $\{ \rho_C \}$ has $\delta$-approximate polynomial circuit complexity; i.e. there exists a polynomial $p(m)$ such that for all $C$, the $\delta$-approximate circuit complexity of $\rho_C$ is at most $p(m)$. We will show that this implies $\QCMA = \QMA$, contradicting our assumption. For the sake of clarity we will write the proof for the case of \emph{exact} circuit complexity (i.e. $\delta = 0$). Generalizing the argument to larger $\delta$ is straightforward.

\newcommand{\sA}{{\sf A}}
\newcommand{\sW}{{\sf W}}

We define the following $\QCMA$ verifier $V$:

\begin{center}
\begin{mdframed}
\textbf{$\QCMA$ verifier $V$} \\
Input: $\QMA$ verifier circuit $C \in L_{yes} \cup L_{no}$, witness string $w \in \{0,1\}^{\ell}$ with $\ell = \poly(n)$. 
\begin{enumerate}
	\item Interpret $w$ as the description of a quantum circuit $W$ acting on $m' \geq m$ qubits.
	\item Generate the $m'$-qubit state $\ket{\varphi}$ formed by running the circuit $W$ on $\ket*{0^{m'}}$. 
	
	\item Let $\rho$ denote the $m$-qubit mixed state obtained by tracing out all but the first $m$ qubits of $\ket{\varphi}$. The $m$ qubits of $\rho$ are divided into an $n$-qubit $\state$ register  and an $(m-n)$-qubit $\time$ register.	
	\item Let $\rho_1$ denote the reduced density matrix of $\rho$ on the $\state$ register.
	\item Let $\rho_2 = D (\rho_1 \otimes \ketbra{0^s}{0^s}_{\sA}) D^\dagger$, where $s$ is the number of ancilla bits used by the decoding procedure $D$, and $\sA$ denotes the label of the ancilla register. Note that $\rho_2$ is a $(n+s)$-qubit state. 
	
	\item Let $\rho_3$ denote the marginal of $\rho_2$ on the first $k$ qubits. %
	\item Measure the first qubit of $C \rho_3 C^\dagger$ in the standard basis, and accept if the outcome is $\ket{1}$. Otherwise, reject. 

\end{enumerate}
\end{mdframed}

\end{center}

\medskip
\noindent We now show that the circuit $V$ decides the language $L$.

\paragraph*{\textbf{Soundness.}} Suppose $C \in L_{no}$. Then for all witnesses $w$, $V$ accepts $(C,w)$ only if measuring the first qubit of $C \rho_4 C^\dagger$ yields outcome $\ket{1}$. However, by definition for all witness states $\ket{\xi}$, $C$ accepts $\ket{\xi}\ket{0}$ with probability at most $\gamma$. Therefore $V$ accepts with probability at most $\gamma$.

\paragraph*{\textbf{Completeness.}} Fix $C \in L_{yes}$. Let $H$ denote the $c$-dimensional Feynman-Kitaev clock Hamiltonian, acting on $m$ qubits, corresponding to $C$. Since $C$ is a {\em yes} instance we have $\lambda_{\min}(H) \leq \gamma$. By assumption, there exists an $\eps$-noisy ground state $\rho$ of $H$ such that $\rho$ has polynomial-size circuit complexity. In other words, there exists a circuit $W$ of size $m' = \poly(m)$ such that $\rho$ is the marginal state of $W \ket*{0^{m'}}$ on the first $m$ qubits. 

We argue that the \emph{classical} description $w$ of the circuit $W$ serves as a witness that $C$ is a \emph{yes} instance. The state $\rho$ created by the verifier in Step $3$ is precisely the $\eps$-noisy ground state. Again for simplicity we present the argument for the case that $\rho$ is an $\eps$-error state of $H$; the argument extends to noisy ground states via linearity.

Before going through the argument in detail, we first describe the high level strategy. The ground space of the Hamiltonian $H$ is spanned by history states of the circuit $C'$. Since $\rho$ is a low-error state, it agrees with a ground state $\sigma$ of $H$ in a $1 - \eps$ fraction of qubits. For sake of exposition, assume that this ground state is a history state of an accepting execution of $C'$ (in general, $\sigma$ could be a convex combination of different history states). The history state is a superposition over snapshots of $C'$, and therefore by construction half of these snapshots belong to the ``idling'' period of the circuit; the $\sR_1 \sR_2$ registers of these snapshots will be a witness state $\ket{\xi}$ encoded with the code $Q$. This implies, informally speaking, that half of the mass of $\rho$ looks like a codeword of $Q$ that has been corrupted in at most $\eps m$ qubits. Since $\eps m$ is smaller than the distance of the code $Q$, this codeword is recoverable by applying the decoding procedure $D$, and thus the verifier $V$ can obtain the original, unencoded witness state $\ket{\xi}$ (with probability half), which is then used as input for the circuit $C$. Since the history state described an accepting execution of $C'$, this implies that $C$ also accepts $\ket{\xi}$, and thus $V$ accepts with probability at least $1/2$.

We now proceed with the argument in detail. Since $\rho$ is an $\eps$-error state of $H$, there exists a density matrix $\pi$ such that $\Tr(H \pi) \leq \gamma$ and $\Tr_S(\rho) = \Tr_S(\pi)$ for some subset $S \subseteq [m]$ of size at most $\eps m$. Let $\pi = \sum_i p_i \ketbra{\Omega_i}{\Omega_i}$. By Markov's inequality, with probability at least $\alpha_{good} \geq 1 - \sqrt{\gamma}$ over $p_{i}$ we have $\bra{\Omega_i} H \ket{\Omega_i} \leq \sqrt{\gamma}$. Call such $i$'s \emph{good}, and \emph{bad} otherwise. By Theorem~\ref{theorem-closegroundstates}, we have that for a good $i$ there exists a history state $\ket{\Psi_i} = \frac{1}{\sqrt{T+1}} \sum_t \ket{\clock_{c}(t)}_\time \otimes \ket{\psi_{i,t}}_\state$ for the circuit $C'$ such that
\begin{equation}
	\| \ketbra{\Omega_i}{\Omega_i} - \ketbra{\Psi_i}{\Psi_i} \|_1 \leq \poly(m) \sqrt{\gamma}.
\end{equation}
Write $\pi = \pi_{good} + \pi_{bad}$ where $\pi_{bad} = \sum_{\text{$i$ bad}} p_i \ketbra{\Omega_i}{\Omega_i}$ and
\begin{align}
	\pi_{good} &= \sum_{\text{$i$ good}} p_i \ketbra{\Omega_i}{\Omega_i} \\
			  &= \sum_{\text{$i$ good}} p_i \ketbra{\Psi_i}{\Psi_i} - \left(\ketbra{\Psi_i}{\Psi_i} - \ketbra{\Omega_i}{\Omega_i} \right) \\
			  &= \left ( \sum_{\text{$i$ good}} p_i \ketbra{\Psi_i}{\Psi_i} \right) + \mu.
\end{align}
where $\| \mu \|_1 \leq \poly(m) \sqrt{\gamma}$. Define $\pi_{err} = \pi_{bad} + \mu$. Observe that $\| \pi_{err} \|_1 \leq \| \mu \|_1 + \| \pi_{bad} \|_1 \leq \poly(m) \sqrt{\gamma}$. Thus,
\begin{equation}
	\pi = \sum_{\text{$i$ good}} p_i \ketbra{\Psi_i}{\Psi_i} + \pi_{err}
\end{equation}
In the verifier $V$, the state $\rho_1$ is defined to be $\Tr_\time(\rho)$. Let $S' = S\setminus \time$. Since $\Tr_S(\rho) = \Tr_S(\pi)$, we have
\begin{align}
	\Tr_{S'} (\rho_1) &= \Tr_{S' \cup \time}(\rho) \\
	 &= \Tr_{S \cup \time}(\pi) \\
	&= \sum_{\text{$i$ good}} p_i \Tr_{S \cup \time}(\ketbra{\Psi_i}{\Psi_i}) + \Tr_{S \cup \time}(\pi_{err}) \\
	&= \sum_{\text{$i$ good}} \frac{p_i}{T+1} \sum_t \Tr_{S'} \left(\ketbra{\psi_{i,t}}{\psi_{i,t}} \right) + \Tr_{S \cup \time}(\pi_{err}).
	\label{eq:rho_1}
\end{align}
Call $t$ \emph{idling} if $T_{enc} < t \leq T_{dec}$, and \emph{active} otherwise. By construction of circuit $C'$, for idling $t$ and good $i$ we have
\begin{equation}
	\ket{\psi_{i,t}} = E \ket{\xi_i}\ket*{0^{n-k}}.
\end{equation}
for some $k$-qubit witness state $\ket{\xi_i}$. Thus it would seem like $\rho_1$ has roughly half probability mass on codewords of $Q$, except we have only established this \emph{after} we trace out the qubits in $S'$. We need to argue that we can decode these codewords by applying the decoding procedure $D$ to $\rho_1$.

\begin{lemma}
\label{lem:purification_decoding}
	Let $S \subseteq [n]$ be a subset of size at most $\beta n/2$. Let $\phi$ be a density matrix on an $n$ qubit register $\state$ such that $\Tr_S(\phi) = q \sum_i p_i \Tr_S(\ketbra{\psi_i}{\psi_i}) + (1 - q) \vartheta$, where $0 \leq q \leq 1$, $\{p_i\}$ is a probability distribution and for all $i$, $\ket{\psi_i} = E \left ( \ket{\xi_i} \ket{0^{n-k}} \right)$ is a codeword of the code $Q$. Let $\Delta = \Tr_{\sA} ( D \phi \otimes \ketbra{0^s}{0^s}_{\sA} D^\dagger)$ be the result of applying the decoding procedure to $\phi$ (along with some ancilla), followed by tracing out the ancilla register. We have then that
	\begin{equation}
		F \left ( \Delta, \Gamma \right ) \geq q.
	\end{equation}
	where $F(\cdot,\cdot)$ is the fidelity function and $\Gamma$ is the $n$-qubit state defined as
	\begin{equation}
		\Gamma = \sum_i p_i \ketbra{\xi_i}{\xi_i} \otimes \ketbra*{0^{n-k}}{0^{n-k}}.
	\end{equation}
\end{lemma}
\begin{proof}
Let $\sigma$ denote $\sum_i p_i \Tr_S(\ketbra{\psi_i}{\psi_i})$, and consider the following purification $\ket{\sigma}$ of $\sigma$:
\begin{align}
	\ket{\sigma} &= \sum_i \sqrt{p_i} \ket{\psi_i}_{\state} \ket{i}_{\sR} 
\end{align}
Let $\ket{\phi}$ be an arbitrary purification of $\phi$ on registers $\state$ and $\sR$. By Uhlmann's Theorem, there exists a unitary $J$ acting on registers $\sR$, and the qubits of $\state$ indexed by $S$ such that the fidelity between $J \ket{\sigma}$ and $\ket{\phi}$ is equal to the fidelity between $\sigma$ and $\Tr_{S} (\phi)$, which is at least $q$. 

Next we have
\begin{align}	
	D J \ket{\sigma}_{\state\, \sR} \ket{0^s}_{\sA}  = \sum_i  \sqrt{p_i} \,  DJ  \ket{\psi_i}_{\state} \ket{i}_{\sR} \ket{0^s}_{\sA}
\end{align}
We can interpret the state $J \ket{\psi_i}_\state \ket{i}_{\sR} $ as at most $\beta n/2$ qubits of $\ket{\psi_i}$ getting entangled with the environment register $\sR$. Since the distance of $Q$ is $\beta n$, this implies that the decoding procedure $D$ can correct for this error:
\begin{equation}
DJ \ket{\psi_i}_{\state}  \ket{i}_{\sR} \ket{0^s}_{\sA} = \ket{\xi_i} \ket{0^{n-k}} \ket{\tau_i}_{\sR \sA}
\end{equation}
where $\ket{\tau_i}$ is a junk state. It is easy to see that the junk states $\ket{\tau_i}$ depend only on $i$, and are independent of $\ket{\psi_i}$; otherwise the decoding procedure would fail on superpositions of codewords. In particular, this implies that the $\{ \ket{\tau_i} \}$ are orthogonal. 

Thus we have
\begin{align}
	\Tr_{\sR \sA} \left ( DJ \ketbra{\sigma}{\sigma} \otimes \ketbra{0^s}{0^s} J^\dagger D^\dagger \right) = \sum_i p_i \ketbra{\xi_i}{\xi_i} \otimes \ketbra*{0^{n-k}}{0^{n-k}} = \Gamma.
\end{align}
Putting everything together, we get
\begin{align}
	F \left (\Delta, \Gamma \right ) &\geq F \left (D \phi \otimes \ketbra{0^s}{0^s} D^\dagger, DJ \ketbra{\sigma}{\sigma} \otimes \ketbra{0^s}{0^s} J^\dagger D^\dagger \right ) \\
	&=  F \left (\phi, J \ketbra{\sigma}{\sigma} J^\dagger \right ) \\
	&\geq q
\end{align}
where the first inequality follows from the fact that the fidelity function is non-increasing under quantum operations, the equality is because the fidelity function is unitarily invariant, and the last inequality follows from our lower bound on the fidelity between $J\ket{\sigma}$ and $\ket{\phi}$.
\end{proof}

Call a pair $(i,t)$ \emph{nice} if both $i$ is good and $t$ is idling. To apply Lemma~\ref{lem:purification_decoding}, we let $q = \sum_{\text{$(i,t)$ nice}} \frac{p_i}{T+1}$ denote the probability of a nice pair $(i,t)$, and observe that $\Tr_{S'} (\rho_1) = q \sum_{\text{$(i,t)$ nice}} p_{i,t}' \ketbra{\psi_{i,t}}{\psi_{i,t}} + (1 - q) \vartheta + \Tr_{S \cup \time}(\pi_{err})$, where $p_{i,t}' = \frac{p_i}{q(T+1)}$. Furthermore, we have that $|S'| \leq |S| \leq \eps m \leq \beta n/2$. Then, Lemma ~\ref{lem:purification_decoding} implies that the state $\rho_3$ obtained by the $\QCMA$ verifier $V$ satisfies
\begin{align}
	F( \rho_3, \sum_{\text{$(i,t)$ nice}} p_{i,t}' \ketbra{\xi_i}{\xi_i}) \geq q - \poly(m) \sqrt{\gamma}
\end{align}
where the error term $\poly(m) \sqrt{\gamma}$ comes from the trace norm of $\Tr_{S \cup \time}(\pi_{err})$. 
Furthermore, $\ket{\Psi_i}$ satisfies 
\begin{equation}
	\bra{\Psi_i} H \ket{\Psi_i} \leq \poly(m) \sqrt{\gamma}
\end{equation}
because it is close to a low-energy state $\ket{\Omega_i}$. This implies that for good $i$ the circuit $C'$ (and therefore the circuit $C$) accepts $\ket{\xi_i,0}$ with probability at least $1 - \poly(m) \sqrt{\gamma}$. Thus the probability that $C$ accepts $\sum_{\text{$(i,t)$ nice}} p_{i,t}' \ketbra{\xi_i}{\xi_i}$ is at least $1 - \poly(m) \sqrt{\gamma}$. 

This implies that the probability that $C$ accepts the state $\rho_3$ is at least $\Omega(q^2)$. Since $i$ is good with probability $1 - \sqrt{\gamma}$ and $t$ is idling with probability $1/2$, this implies that $q \geq 1/2 - o(1)$, so the probability that $V$ accepts input $(C,w)$ is at least $1/4 - o(1)$, which is bounded away from $\gamma$. This implies that $V$ correctly decides $L$ and therefore $\QCMA = \QMA$, contradicting our assumption.

\end{proof}

\subsection{Unconditional constructions of SNGS Hamiltonians from oracle separations}
\label{unconditionalsuperpoly}

In the previous section, we showed the existence of a family of SNGS Hamiltonians, assuming that $\QCMA \neq \QMA$. In this section we \emph{unconditionally} show the existence of a family of SNGS Hamiltonians, with the tradeoffs that (1) the locality of the Hamiltonians is superlogarithmic, and (2) the Hamiltonians are not fully explicit: one of the terms of the Hamiltonian is chosen via the probabilistic method.

We accomplish this by leveraging the oracle separation results of~\cite{aaronson2007quantum,DBLP:journals/corr/FeffermanK15} that show $\QCMA^{\Oo} \subsetneq \QMA^{\Oo}$ relative to an oracle $\Oo$. The oracle given by~\cite{aaronson2007quantum} is a unitary oracle that reflects around some state; i.e. it is of the form $\II - 2\ketbra{\psi}{\psi}$ for some state $\ket{\psi}$. The oracle given by~\cite{DBLP:journals/corr/FeffermanK15} is slightly more structured: $\Oo$ is a permutation matrix in the standard basis. In this section we will use the permutation oracle of~\cite{DBLP:journals/corr/FeffermanK15}.

We will be more precise about what we mean by $\QMA^\Oo$ and $\QCMA^\Oo$. The oracle $\Oo$ is actually a countably infinite sequence of unitary operators $\{ \Oo_1,\Oo_2,\ldots \}$ where each $\Oo_m$ acts on $m$ qubits. We say that a language $L$ is in $\QMA^\Oo$ if there exists a uniform family of $\poly(n)$-size $\QMA^\Oo$-verifier circuits $\{ V_n \}$ that decides $L$. A $\QMA^\Oo$-verifier circuit $V_n$ is the same as a $\QMA$-verifier circuit except $V_n$ can use the unitary $\Oo_{m(n)}$ as an elementary gate (the ``size'' of the oracle that $V_n$ is allowed to query is determined by some function $m(n)$ that grows at most polynomially in $n$).  The complexity class $\QCMA^\Oo$ has an analogous definition.

\begin{theorem}[\cite{DBLP:journals/corr/FeffermanK15}]
\label{theorem-fk}
There exists a family of oracles $\{ \Oo_n \}$ and a unary language $L = (L_{yes},L_{no})$ such that 
	\begin{enumerate}
		\item $L \in \QMA^\Oo$ but $L \notin \QCMA^\Oo$.
		\item The function $m(n) = O(\log^{1 + \kappa} n)$ for arbitrarily small $\kappa > 0$.
		\item There exists a uniformly generated family of $\QMA^\Oo$-verifier circuits $\{ V_n \}$ such that $V_n$ makes exactly one call to $\Oo_{m(n)}$. 		
		\item The oracles $\Oo_m$ are permutation matrices in the standard basis.
	\end{enumerate} 
\end{theorem}

At a high level the oracle separations of~\cite{aaronson2007quantum,DBLP:journals/corr/FeffermanK15} are proved in the following manner. For concreteness we discuss the separation proved by~\cite{DBLP:journals/corr/FeffermanK15}. Fix an $n$, and let $m = m(n)$ denote the number of qubits the permutation oracle $\Oo_m$ acts on. The oracle $\Oo_m$ corresponds to a permutation $\sigma_m$ on $\{0,1\}^m$. The language $L$ corresponds to the following problem: given $n$, decide whether $\sigma_m^{-1}( \{0^{m/2}\} \times \{0,1\}^{m/2})$ (i.e. the preimage of strings that start with $m/2$ leading $0$'s) has at least $2/3$ fraction of strings with the least significant bit being $0$ (\emph{yes} case), or at most $1/3$ fraction with the least significant bit being $0$ (\emph{no} case), promised that one is the case\footnote{An important part of the oracle model of~\cite{DBLP:journals/corr/FeffermanK15} is that the verifiers are \emph{not} given access to the inverse oracle $\Oo^\dagger$.}. This problem can be solved easily in $\QMA^\Oo$: Given a witness state $\ket{\xi}$, the $\QMA^\Oo$ verifier performs one of two tests with half probability each: either (1) apply the oracle $\Oo$ on $\ket{\xi}$ and measure to check that the resulting state is the equal superposition over strings of the form $0^{m/2}x$ for $x \in \{0,1\}^{m/2}$; or (2) measure the least significant bit of $\xi$ and check whether it is $0$.

In the \emph{yes} case, the prover can provide the witness state that is the equal superposition
\begin{equation}
	\frac{1}{2^{m/4}} \sum_{x \in \sigma_m^{-1}(\{0^{m/2}\} \times \{0,1\}^{m/2})} \ket{x}.
	\label{eq-subset-state}
\end{equation}
This will convince the verifier with probability $2/3$. In the \emph{no} case, either the witness state is not a correct preimage state, or if it is, the least significant bit will not be $0$ with high probability. Either way, the verifier will reject with good probability.

The more difficult part is to construct $L$ and $\Oo$ so that $L \notin \QCMA^\Oo$. Fefferman and Kimmel use the probabilistic method to show the existence of $L$ and $\Oo$ such that any verifier circuit $V$ --- even when given a witness string --- requires at least $2^{\delta m}/\poly(n)$ queries to $\Oo$ to decide $L$, for some constant $\delta > 0$. This is superpolynomial in $n$ when $m = \omega(\log n)$, so thus no $\QCMA^\Oo$ verifier can decide $L$.

In the next theorem, we construct a family of Hamiltonians whose ground states essentially encode states of the form in~\eqref{eq-subset-state}, by applying the proof technique of Theorem~\ref{theorem-superpoly} to the oracle verifier for $L$. This will require embedding the oracle $\Oo$ in the Hamiltonian, which is what worsens the locality from constant to slightly superlogarithmic. On the other hand, the Hamiltonian family is now unconditionally SNGS.

\superpolyoracle* %

\begin{proof}
The proof follows identically to that of Theorem~\ref{theorem-superpoly}. Let $L$ be the unary language from Theorem~\ref{theorem-fk}. Call the verifier circuit $C_n$ a \emph{yes} instance if $1^n \in L$, and otherwise call it a \emph{no} instance. 

For every \emph{yes} instance $C_n$, apply the circuit transformation $C_n \mapsto C_n'$ as in Theorem~\ref{theorem-superpoly}, and let $H^{(n)}$ be the Feynman-Kitaev Hamiltonian corresponding to $C_n'$. Everything is the same as before, except now $C_n$ has one call to the oracle $\Oo_{m(n)}$, which is not two-local. This means that the term in $H_\text{prop}$ that corresponds to this oracle call will have locality $m(n)$. 

We will let our family of Hamiltonians be $\{ H^{(n)} \}$. If there was a family of noisy ground states $\{ \rho_n \}$ for $\{ H^{(n)} \}$ with polynomial circuit complexity, then we would actually have a $\QCMA^\Oo$ verifier for $L$, contradicting Theorem~\ref{theorem-fk}.

\end{proof}

A couple of remarks are in order:
\begin{enumerate}
	\item This family of local Hamiltonians is almost fully explicit except for a single term in each $H^{(n)}$ that corresponds to the oracle call to $\Oo_{m(n)}$ at some time $t$. Here, the term looks like
\begin{equation}
H_t = \half \left ( - A_{t,t-1} \otimes \Oo_{m(n)}- A_{t-1,t} \otimes \Oo_{m(n)}^\dagger + A_{t,t}\otimes \II + A_{t-1,t-1} \otimes \II \right)
\end{equation}
like any other term in $H_\text{prop}$. The oracle $\Oo$ is chosen randomly; hence the name ``semi-probabilistic construction.''

	\item It is not too hard to argue that \emph{ground states} of $\{ H^{(n)}\}$ require superpolynomial circuit complexity: there are roughly $2^{2^{\Omega(m(n))}}$ different choices of $\Oo_{m(n)}$, but only $2^{\poly(n)}$ states with polynomial circuit complexity. If $m(n)$ is superlogarithmic, then by a counting argument there exists a choice of $\Oo_{m(n)}$ such that the corresponding Hamiltonian $H^{(n)}$ has a ground state that cannot be described using polynomial-size circuits. 
	
	Theorem~\ref{theorem-superpoly-oracle} goes further than this statement; it says that all \emph{noisy ground} states (which includes all low-error states) of $\{ H^{(n)} \}$ require superpolynomial circuit complexity. Here, the straightforward counting argument breaks down. This is because there are potentially at least $2^{2^{\eps n}}$ $\eps$-error states for $H^{(n)}$: if $\Tr_S(\ketbra{\Phi}{\Phi}) = \Tr_S(\ketbra{\Psi}{\Psi})$ for some ground state $\ket{\Psi}$, then $\ket{\Phi} = U \otimes \II \ket{\Psi}$ for some $\eps n$-local unitary $U$, and there are roughly $2^{2^{\eps n}}$ such unitaries. This a much greater number than $2^{2^{m(n)}}$ for sublinear $m(n)$.
\end{enumerate} %

\section{Asymptotically good approximate qLWC codes}
\label{sec-good-approx-qldpc-codes}

In this section, we show how to obtain a family of approximate qLWC codes with constant relative distance, constant rate, constant locality, and polynomially small error. This construction is a distillation of the technique used to prove Theorem~\ref{theorem-superpoly}.

\goodapproxqldpc* %

The proof of this is similar to that of Theorem \ref{theorem-superpoly}. Let $V$ be the encoding circuit of an error-correcting code of good rate. Let $C$ be a circuit consisting of running $V$ followed by applying identity gates polynomially many times. We claim then that the Feynman-Kitaev Hamiltonian $H$ corresponding to the circuit is an approximate qLWC code. By tracing out the $\time$ register of any state $\mathcal{E} \circ \Enc(\rho)$ for $\rho \in D((\CC^q)^{\otimes k})$ and then applying error correction, we can generate a state close to $\rho$ in trace distance. \\

\begin{proof}
Let $\{ Q_k \}$ be an explicit family of quantum error-correcting codes such that each $Q_k$ is an $[[n,k,d]]_q$ code with $q=O(1)$, $n = O(k)$ and $d \geq \theta n$. Such asymptotically good codes are known for sufficiently small $\theta$; furthermore, they can be assumed to be CSS codes (see, e.g.,~\cite{ashikhmin2001asymptotically}). 

Fix a $k$ and the code $Q = Q_k$. Fix $\delta = \delta(n)$. Let $V$ be the unitary encoding circuit for $Q$: for an $k$-qudit message $\ket{\psi}$, it maps $\ket{\psi} \ket{0}$ to an $n$ qudit codeword. Since $Q$ is an explicit CSS code, $V$ can be computed by a circuit of size $T_V \leq n (n - k) \leq n^2$ (see, e.g., ~\cite{gottesman_1997}).

Consider a circuit $C$ consisting of running the encoding circuit $V$ followed by $K = 4T_V/\delta^2$ identity gates so that $K$ satisfies
\begin{equation}
\frac{K}{T_V + K} \geq 1 - \delta^2/4.
\end{equation}
Therefore, the ``waiting period'' is at least $1 - \delta^2/4$ fraction of the total circuit running time.

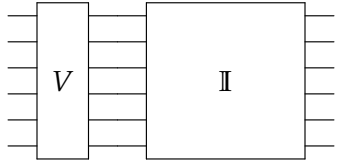
\begin{figure}[h!]
\centering
$
\Qcircuit @C=1em @R=0em {
& \multigate{5}{V} & \qw & \multigate{5}{\qquad \II \qquad} & \qw\\
& \ghost{V} & \qw & \ghost{\qquad \II \qquad} & \qw\\
& \ghost{V} & \qw & \ghost{\qquad \II \qquad} & \qw\\
& \ghost{V} & \qw & \ghost{\qquad \II \qquad} & \qw\\
& \ghost{V} & \qw & \ghost{\qquad \II \qquad} & \qw\\
& \ghost{V} & \qw & \ghost{\qquad \II \qquad} & \qw\\
}
$
\end{figure}

The circuit size $T_C$ of $C$ is equal to $T_V + K = (1 + 4/\delta^2)T_V$. Define
\begin{equation}
	r = \left \lceil \frac{\log T_C}{\log n} \right \rceil.
\end{equation}
For example, if $\delta^2 = n^{-a}$, then $T_C = O(n^{2 + a})$, so $r = 2 + a$. Let $H = H_\text{in} + H_\text{prop} + H_\text{stab}$ be the Feynman-Kitaev Hamiltonian of $C$ with a $r$-dimensional clock, except we omit the $H_\text{out}$ term and the term $H_{in}$ enforces that the input to the circuit has the last $n - k$ qudits set to $\ket{0}$. %
The terms $H_\text{prop}$ and $H_\text{stab}$ are the same as in Section~\ref{sec-clock-construction}.  We have that $H$ has locality $3 + 2r$ and acts on $m = r \lceil T_C^{1/r} \rceil + n \leq (r+1)n$ qudits.

It is easy to verify that the ground energy of $H$ is $0$ and the ground space is spanned by history states of the circuit $C$ where the initial state is of the form $\ket{\xi} \otimes \ket{0}^{\otimes (n-k)}$ for some $k$-qudit state $\ket{\xi}$.

We will let $H$ be the code Hamiltonian of our approximate qLWC code.

\paragraph*{Encoding map.} The encoding map $\Enc$ which acts on $k$-qudit density matrices $\rho$ acts as follows:
\begin{equation}
\Enc(\rho) = W(\rho \otimes \ketbra{0}{0}^{\otimes (n-k)}) W^\dagger.
\end{equation}
where $W$ is the following $n$-qudit isometry:
\begin{equation}
\ket{\psi}_\state \overset{W}{\mapsto} \frac{1}{\sqrt{T_C+1}} \sum_{t = 0}^{T_C} \ket{\clock_r(t)}_\time \otimes C_t \cdots C_1 \ket{\psi}_\state.
\end{equation}

The isometry $W$ has an efficient circuit: first, it initializes the $\time$ register into an equal superposition of clock states
\begin{equation}
	\frac{1}{\sqrt{T_C+1}} \sum_{t = 0}^{T_C} \ket{\clock_r(t)}_\time.
\end{equation}
Then, conditioned on time $\ket{\clock_k(t)}$ in the $\time$ register, $W$ applies the unitary $C_t \cdots C_1$ to the $\state$ register.

Since $\Enc$ generates history states for the circuit $C$, we have that $\Tr(H \Enc(\rho)) =0$ for all $\rho$. Conversely, if $\ket{\Psi}$ is a state such that $\expval{H}{\Psi} = 0$, then $\ket{\Psi}$ is a history state of the circuit $C$ with initial state $\ket{\xi} \otimes \ket{0}$. Thus we have $\ketbra{\Psi}{\Psi} = \Enc(\ketbra{\xi}{\xi})$.

\paragraph*{Approximate decoding.} Let $\tilde V$ be the decoding circuit expressed as a unitary mapping an encoded state into two registers, $\state$ and $\ancilla$. Define the decoding map $\Dec(\cdot)$ as
\begin{equation}
\Dec(\sigma) \defeq \Tr_\ancilla(\tilde V \Tr_\time(\sigma) {\tilde V}^\dagger).
\end{equation}
Here $\ancilla$ are the last $n - k$ qudits of the $\state$ register. Since $\tilde V$ is a polynomial-size circuit (as previously discussed), the map $\Dec$ also has polynomial-size circuits.

Let $\ket{\phi} \in (\CC^q)^{\otimes k} \otimes \mathcal{R}$ be a $k$-qudit message $\rho$ that has been purified (i.e., $\rho = \Tr_{\mathcal{R}}(\ketbra{\phi}{\phi})$). Let a Schmidt decomposition of $\ket{\phi}$ be $\sum_i \sqrt{p_i} \ket{\xi_i} \ket{i}$, where the $\{ \ket{\xi_i} \}$ correspond to the Hilbert space $(\CC^q)^{\otimes k}$ and the $\{ \ket{i} \}$ are orthonormal vectors in $\mathcal{R}$. Let $\ketbra{\Psi}{\Psi} = \Enc(\ketbra{\phi}{\phi})$, so that $\ket{\Psi} = \sum_i \sqrt{p_i} \ket{\Psi_i} \ket{i}$ where 
$\ket{\Psi_i} = \frac{1}{\sqrt{T_C + 1}} \sum_t \ket{\clock_r(t)} \otimes \ket{\psi_{i,t}}$ is the history state for circuit $C$ on input state $\ket{\xi_i}\otimes \ket{0}^{\otimes n-k}$. Recall for $t$ in the ``waiting period'' ($t > T_V$) we have $\ket{\psi_{i,t}} = \ket{\Gamma_i} \defeq V \ket{\xi_i}\ket{0}^{\otimes (n-k)}$. Thus we can write
\begin{equation}
	\ket{\Psi_i} = \frac{1}{\sqrt{T_C + 1}} \sum_{t = 0}^{T_V} \ket{\clock_r(t)} \otimes \ket{\psi_{i,t}} + \underbrace{\left (\frac{1}{\sqrt{T_C + 1}} \sum_{t = T_V + 1}^{T_C} \ket{\clock_r(t)} \right)}_{\defeq \ket{\tau}} \otimes \ket{\Gamma_i}.
\end{equation}
Define $\ket*{\wt{\Psi_i}} \defeq \sum_i \sqrt{p_i} \ket{\tau} \otimes \ket{\Gamma_i} \otimes \ket{i}$. Note that $\ket*{\wt{\Psi_i}}$ has norm equal to $\| \ket{\tau} \|^2 \geq 1 - \delta^2/4$ because waiting period is at least $1 - \delta^2/4$ of the total time. Furthermore, $| \ip*{\wt{\Psi}_i}{\Psi_i}|^2 = \| \ket{\tau} \|^2$. Using the relation between the trace distance and fidelity between two pure states, we have
\begin{equation}
	\left \| \ketbra{\Psi}{\Psi} - \ketbra*{\wt{\Psi}}{\wt{\Psi}} \right \|_1 \leq \delta.
\end{equation}
where $\ket*{\wt{\Psi}} = \sum_i \sqrt{p_i} \ket*{\wt{\Psi_i}} \ket{i}$.
Let $\mathcal{E}$ be a completely positive, trace preserving map acting on at most $(d - 1)/2$ qudits of $\sigma$. Since $Q$ is a code that can correct up to $(d-1)/2$ errors, and $\ket*{\wt{\Psi}}$ is a (sub-normalized) superposition of codewords of $Q$ (along with a time state $\ket{\tau}$ that gets traced out by $\Dec$), we have that $\Dec \circ \mathcal{E}(\ketbra*{\wt{\Psi}}{\wt{\Psi}}) = \ketbra{\phi}{\phi}$. Since the trace distance is non-increasing under quantum operations, we have that
\begin{align}
\left \| \Dec \circ \mathcal{E} \circ \Enc(\ketbra{\phi}{\phi})- \ketbra{\phi}{\phi} \right \|_1 
 &= \left \| \Dec \circ \mathcal{E}(\ketbra{\Psi}{\Psi}) - \Dec \circ \mathcal{E}(\ketbra*{\wt{\Psi}}{\wt{\Psi}}) \right \|_1 \\
 &\leq \left \| \ketbra{\Psi}{\Psi} - \ketbra*{\wt{\Psi}}{\wt{\Psi}} \right \|_1 \\
 &\leq \delta.
\end{align}

\paragraph*{Parameters.}
The parameters of our approximate qLWC code Hamiltonian $H$ are as follows: for message size $k$ and error $\delta$, we get
\begin{center}
\begin{tabular}{ r c c l }
 Qudit dimension &$q$& $=$ &$O(1)$,\\ 
 Locality & $w$ & $=$ & $3 + 2r$,\\
 Blocklength & $m$ & $\leq$ & $(r+1)n \leq O(rk)$, \\
 Distance & $d$ & $=$ & $\theta n \geq \frac{\theta}{(r+1)}m$   
\end{tabular}
\end{center}
where $r = \log (1 + 4/\delta^2)/\log n + 2$.
\end{proof}

\medskip

We conclude with a few remarks about our qLWC construction. The code Hamiltonian has the following properties:
\begin{enumerate}
	 \item It is a \emph{non-commuting} Hamiltonian --- whereas usually most code Hamiltonians, such as those coming from stabilizer codes, are commuting.
	 \item The terms of the Hamiltonian (i.e. the checks) do not behave as syndrome measurements as in stabilizer codes; they mainly describe groundspace, but are not directly used in the decoding procedure.
	 \item It is gapless: its gap decays as $\min \{ n^{-2}, \poly(\delta) \}$. 
	 \item The number of terms that act on any given particle (called the \emph{degree}) grows as a function of $n$. In contrast, qLDPC codes have bounded degree. 
	 \item The output state after applying an error and decoding is a mixed state with $\geq (1-\delta)$ probability of being the original state:
	 \begin{equation}
	 \Dec \circ \mathcal{E} \circ \Enc(\ketbra{\phi}{\phi}) = (1-\delta) \ketbra{\phi}{\phi} + \delta \rho_{\junk}.
	 \end{equation}
	 This is in contrast to the general definition of an approximate qLWC code where the output need only be close in trace distance.
\end{enumerate}

Furthermore, we note that our qLWC code construction is inherently quantum. The analagous classical construction of our code would use the Cook-Levin tableau \cite{Cook:1971:CTP:800157.805047,Levin:73} instead of the Feynman-Kitaev clock Hamiltonian. The rows of the resulting tableau would simply be copies of a codeword in a good error-correcting code, or equivalently, a repetition code applied to a good error-correcting code. This would not preserve the good rate and distance properties. Whereas, our qLWC code maintains the rows of the tableau in \emph{superposition}, which is essential for our construction to have good rate and distance parameters.

Lastly, we note that our qLWC construction can likely be improved by using more sophisticated clock Hamiltonian constructions; for example, the techniques of~\cite{caha2017feynman,1609.08571} may be useful for reducing the number of clock particles, reducing the locality of the Hamiltonian terms, and improving the spectral gap.

\bibliography{references}
\bibliographystyle{alpha}

\appendix

\section{The $k$-dimensional Feynman-Kitaev clock construction}
\label{appendix-kitaevclock}

Let $C$ be a circuit consisting of $T$ two-local acting on $n$ qubits. Suppose that $T \leq n^k$. The usual Feynman-Kitaev clock construction using a unary clock (e.g., as presented in Section~\ref{sec-clock-construction}) would produce a local Hamiltonian acting on $O(T)$ qubits. Here, we present an alternate construction where the Hamiltonian acts on $O(kn)$ qubits, at the cost of increasing the locality of the Hamiltonian by a constant. 

Let $d = \lceil T^{1/k}\rceil + 1 = O(n)$. The $\time$ register is comprised of $k$ ``subclocks'' denoted $\Rr_1, \Rr_2, \ldots, \Rr_k$ where each register $\Rr_i$ consists of $d$ qubits. We define
\begin{equation}
\ket{\clock_k(t)}_\time = \bigotimes_{i=1}^k \ket{\unary(a_i)}_{\Rr_i}
\end{equation}
where $a_1, \ldots, a_k \in \{0, \ldots, d-1\}$ is the unique solution to $t = \sum_{i=1}^{k} a_i d^{i-1}$. The $\state$ register consists of $n$ qubits. 

The modified Hamiltonian $H$ is the sum of four local terms $H_\text{in} + H_\text{prop} + H_\text{out} + H_\text{stab}$ as before, but we have to take into account this new encoding of the clock register. We give a brief overview of the changes needed to the canonical definition of these terms.

The Hamiltonian $H_\text{in}$ checks that when the $\time$ register is $\ket{\clock_k(0)}$, the $\state$ register is a witness state followed by ancillas. The Hamiltonian $H_\text{out}$ checks that when the $\time$ register is $\ket{\clock_k(T)}$, the $\state$ register's first qubit is $\ket{1}$; equivalently, that the circuit has accepted the computation. The Hamiltonian $H_\text{prop} = \sum_{t=1}^T H_t$ where $H_t$ verifies that the gate $C_t$ was properly applied. However, as described the Hamiltonians would need to encompass the entirety of the $\time$ register in addition to up to two qubits of the $\state$ register. We notice that the `majority' of the qubits in these Hamiltonian terms are dedicated to ensuring that the state is a clock state; these can be separately considered allowing the Hamiltonian to only look at the `pertinent' qubits of the time register. These extra enforcing Hamiltonian terms form the term $H_\text{stab}$. In particular, the stabilizer terms only need to enforce that each register $\Rr_i$ looks unary; meaning, that it looks like $\ket{0, \ldots, 0, 1, \ldots, 1}$. This is sufficient to restrict us to the space of only clock states. The Hamiltonians $H_\text{in},  H_\text{prop}, H_\text{out}, H_\text{stab}$ are defined as follows.

For $i \in \{1, \ldots, k\}$ and $j \in \{0, \ldots, d - 1\}$ let
\begin{equation}
\ket{v_{i,j}} \defeq 
	\begin{cases}
		\II \otimes \ket{0}_{\Rr_i(0)} & j = 0 \\
		\II \otimes \ket{0}_{\Rr_i(j+1)} \otimes \ket{1}_{\Rr_i(j)} & 0 < j < d - 1 \\
		\II \otimes \ket{1}_{\Rr_i(d-1)} & j = d - 1
	\end{cases}
\end{equation}
where we assume the identity acts on all unspecified qubits. Due to the inclusion of $H_\text{stab}$, $\ketbra{v_{i,j}}{v_{i,j}}$ is \emph{effectively} a projector onto the space of values where $\Rr_i$ has value $\ket{\unary(j)}$. For a given $t = \sum_i a_i d^{i - 1}$, define
\begin{equation}
A_t \defeq \ketbra{u_t}{u_t}, \qquad \ket{u_t} \defeq \bigotimes_{i = 1}^k \ket{v_{i,a_i}}.
\end{equation}
Then $A_t$ is \emph{effectively} a projector onto the space where the $\time$ register holds $\ket{\clock_k(t)}$. It is easy now to express the terms of the Hamiltonian.

\begin{enumerate}

\item
	\begin{equation}
		H_\text{in} = \sum_{j = m + 1}^n A_0 \cdot \Pi_{\state(j)}^{(1)}.
	\end{equation} 

\item
	\begin{equation}
		H_\text{prop} = \sum_{t = 0}^T H_t
	\end{equation}
	where for all $t$,
	\begin{equation}
		H_t = \half \left(- \ketbra{u_t}{u_{t-1}} \otimes C_t -\ketbra{u_{t-1}}{u_{t}} \otimes C_t^\dagger + (A_t + A_{t-1}) \otimes \II \right).
	\end{equation}

\item
	\begin{equation}
		H_\text{out} = A_T \otimes \Pi_{\state(1)}^{(0)}.
	\end{equation}

\item
	\begin{equation}
		H_\text{stab} = \sum_{i = 1}^k \sum_{j = 1}^{d-1} \Pi_{\Rr_i(j-1)}^{(0)} \Pi_{\Rr_i(j)}^{(1)}.
	\end{equation}
\end{enumerate}

The projector $\ketbra{v_{i,j}}{v_{i,j}}$ acts on at most 2 qubits. And $\ketbra{v_{i,j}}{v_{i,j-1}}$ act on 3 distinct qubits when both $0 < j - 1$ and $j < d - 1$ and acts on 2 distinct qubits otherwise. A simple ``carrying'' argument shows that $\ketbra{u_t}{u_{t-1}}$ acts on at most $2k + 1$ qubits. Clearly, $A_t$ acts on at most $2k$ qubits. Therefore, all Hamiltonian terms act on at most $2k + 3$ qubits ($2k + 1$ qubits from the $\time$ register and $2$ qubits from the $\state$ register). The total number of qubits acted upon by $H$ is $O(kn)$. 

\section{Proof of Proposition \ref{prop-lightcone}}
\label{appendix-prop-lightcone}

\lightconeprop* %

\begin{proof}
Let $K(A) = K_U(A)$ and $K(B) = K_U(B)$ be the lightcones of $A$ and $B$ with respect to $U$.
Let $L_1, L_2, \ldots, L_d$ denote the layers of gates of the circuit $U$. In other words $U = L_d L_{d-1} \cdots L_1$ where each $L_j$ is a tensor product of two-local gates. In what follows, we will abuse notation and interchangeably treat $L_j$, $K(A),K(B)$, $E(A),E(B)$ as both sets of gates as well as tensor products of two-local unitaries. Furthermore, for notationally simplicity, assume the notation $X_{[1..k]} = X_1 X_2 \ldots X_k$ and analagously $X_{[k..1]} = X_k X_{k-1} \ldots X_1$.

For each layer $j$, let $K_j = (K(A) \cup K(B)) \cap L_j$ (i.e. the intersection of the lightcones with layer $j$). Let $F_j = L_j \setminus K_j$. That is, $F_j$ is the set of gates in layer $j$ that are outside of the lightcones of $A$ and $B$. Suppose that $\rho = \Tr_Y(U \ketbra{0}{0}^{\otimes n} U^\dagger)$ for some set of qudits $Y$.  We have then that
\begin{align}
	\Tr(A \otimes B \rho) &= \Tr \left (  (A \otimes B) U \ketbra{0}{0}^{\otimes n} U^\dagger \right) \\
	&= \bra{0^n} U^\dagger (A \otimes B) U  \ket{0^n} \\
	&= \bra{0^n} L_{[d\ldots 1]}^\dagger (A \otimes B) L_{[d\ldots 1]}  \ket{0^n}
\end{align}

Consider the last layer $L_d = F_d \cup K_d$. Since $F_d$ does not intersect either lightcone (and therefore does not intersect $A \otimes B$), we have that 
\begin{align}
	L_d^\dagger (A \otimes B) L_d &= (F_d^\dagger \otimes K_d^\dagger) (A \otimes B) (F_d \otimes K_d) \\
								 &= K_d^\dagger (A \otimes B) K_d.
\end{align} 
Continuing in this manner, we get that
\begin{align}
	\Tr(A \otimes B \rho) &= \bra{0^n} (K(A) \otimes K(B))^\dagger (A \otimes B) (K(A) \otimes K(B)) \ket{0^n} \\
	&= \bra{0^n} (K(A)^\dagger A K(A)) \otimes (K(B)^\dagger B K(B)) \ket{0^n} \\
	&= \left (\bra{0^n} (K(A)^\dagger A K(A)) \ket{0^n} \right) \left ( \bra{0^n} (K(B)^\dagger B K(B)) \ket{0^n} \right) \\
	&= \left (\bra{0^n} (U^\dagger A U) \ket{0^n} \right) \left ( \bra{0^n} (U^\dagger B U) \ket{0^n} \right) \\
	&= \Tr(A \rho) \cdot \Tr(B \rho).
\end{align}
\end{proof}

\section{A family of asymptotically good quantum codes with efficient encoding and decoding}
\label{sec:good_code}

For completeness, we present a family of asymptotically good quantum codes with efficient encoding and decoding. In an earlier version of this paper, we claimed that the algebraic geometry codes of~\cite{ashikhmin2001asymptotically} formed such a family. While these algebraic geometry codes have asymptotically good parameters (i.e., constant rate, constant relative distance and constant sized alphabets), to our knowledge it has not been established whether these codes have efficient decoding procedures.

Instead, we present an alternate code construction. The idea is to take the \emph{quantum Reed-Solomon code}, which attains the optimal tradeoff between rate and distance\footnote{This optimal tradeoff is also known as the \emph{quantum Singleton bound}.}, but has an alphabet size that grows linearly with the blocklength:

\begin{theorem}[Quantum Reed-Solomon codes]
\label{thm:qrsc}
	There exists an infinite family of $[[N,K,D]]_N$ codes over $N$-dimensional qudits where $K = N/2$ and $D = \frac{N}{4} + 1$. Furthermore, this family admits polynomial-time encoding and decoding. 
\end{theorem}

We remark that the notation $[[N,K]]_q$ for a quantum code $\mathcal{C}$ indicates that we can encode $K$ logical qudits of dimension $q$ into $N$ physical qudits. In other words, the code $\mathcal{C}$ can encode $K \log q$ qu\emph{bits}.

 To reduce the alphabet size of the quantum Reed-Solomon code, we \emph{concatenate} the each symbol of the Reed-Solomon code with a binary inner code that also achieves constant rate and constant relative distance. Fix an $[[N,K,D]]_N$ quantum Reed-Solomon code. Since the alphabet size is $N$, we can use (for example) the $[[n,k,d]]$ binary algebraic geometry codes of~\cite{ashikhmin2001asymptotically} for the inner code, with $n = O(\log N)$, $k = \log N$, and $d = \Omega(\log N)$. 

The codewords of the concatenated code are formed by first encoding a $K\log N$-qubit message into a quantum Reed-Solomon codeword of $N$ physical qudits, and then encoding each of the physical qudits using the $[[n,k,d]]$ inner binary code, to obtain an $Nn$-qubit codeword. The resulting code is binary, encodes $Kn$ qubits  into $Nn$ qubits, and has distance $Dd$. Thus the rate is $K/N = \frac{1}{2}$ and the relative distance is $\Omega(D/N)$, which is a constant. 

We now argue that this code is also efficiently encodable and decodable. The encoding process was described in the previous paragraph; the outer encoding is efficient because of Theorem~\ref{thm:qrsc}, and the inner encoding is efficient because the algebraic geometry codes are efficiently encodable. 

To decode, one can first apply the inner decoding procedure in parallel on each block of $n$ qubits, and then applying the outer decoding procedure on the resulting $N$-qudit message (where we interpret $k$ qubits as an $N$-dimensional qudit). The outer decoding procedure is efficient because the quantum Reed-Solomon code is efficiently decodable; the inner decoding procedure can be done in $\poly(N)$ time, because we can run exhaustive search to perform minimum distance decoding for each of the blocks of $n = O(\log N)$ qubits, with respect to the inner code. This concatenated decoding procedure can correct errors of length up to $Dd$.  

Put together, this implies the following theorem:

\begin{theorem}
\label{thm:good_codes}
	There exists an infinite family of binary quantum codes with constant rate, constant relative distance, and efficient encoding and decoding procedures.
\end{theorem}

\end{document}